\documentclass[pra, twocolumn, floatfix, superscriptaddress, longbibliography, groupedaddress]{revtex4-1}

\usepackage{amsmath, amssymb}
\usepackage[dvipsnames]{xcolor}
\usepackage{color, graphicx}
\usepackage{hyperref}
\usepackage{amsthm}

\usepackage{cancel} %for strikeout
\usepackage[normalem]{ulem}

\newcommand{\ket}[1]{{| #1 \rangle}}
\newcommand{\bra}[1]{{\langle #1 |}}

\newcommand{\expect}[1]{{\langle #1 \rangle}}

\newcommand{\bs}[1]{{\boldsymbol #1}}
\newcommand\id{\mathbb{I}_n}
\newcommand\bigz{\boldsymbol{0}}

\newcommand \ga{G_1}
\newcommand \gb{G_2}

\newcommand{\eqnref}[1]{Eq.~\eqref{#1}}
\newcommand{\appref}[1]{Appendix~\ref{#1}}
\newcommand{\subfiglabel}[1]{{(#1)}}
\newcommand{\figref}[1]{Fig.~\ref{#1}}
\newcommand{\subfigref}[2]{\figref{#1}\subfiglabel{#2}}
\newcommand{\secref}[1]{Sec.~\ref{#1}}

\newtheorem{theorem}{Theorem}
\newtheorem{lemma}[theorem]{Lemma}

\usepackage{mathdots} %for anti diagonal dots

\begin{document}

\title{Degeneracies and symmetry breaking in pseudo-Hermitian matrices}

\author{Abhijeet Melkani}
\email{amelkani@uoregon.edu}
\affiliation{Department of Physics, University of Oregon, Eugene, Oregon, USA}
\date{\today}

\begin{abstract}
Real eigenvalues of pseudo-Hermitian matrices, such as real matrices and $\mathcal{PT-}$symmetric matrices, frequently split into complex conjugate pairs. This is accompanied by the breaking of certain symmetries of the eigenvectors and, typically, also a drastic change in the behavior of the system. In this paper, we classify the eigenspace of pseudo-Hermitian matrices and show that such symmetry breaking occurs if and only if eigenvalues of opposite kinds collide on the real axis of the complex eigenvalue plane. This enables a classification of the disconnected regions in parameter space where all eigenvalues are real---which correspond, physically, to the stable phases of the system. These disconnected regions are surrounded by exceptional surfaces which comprise all the real-valued exceptional points of pseudo-Hermitian matrices. The exceptional surfaces, together with the diabolic points created by their intersections, comprise all points of pseudo-Hermiticity breaking. In particular, this clarifies that the degeneracy involved in symmetry breaking is not necessarily an exceptional point. We also discuss how our study relates to conserved quantities and derive the conditions for when degeneracies caused by external symmetries are susceptible to thresholdless pseudo-Hermiticity breaking. We illustrate our results with examples from photonics, condensed matter physics, and mechanics.
\end{abstract}

\maketitle

\section{Introduction}

Linear operators, such as those representable by matrices, are ubiquitous in physics forming many exact models of nature. They also occur as effective models when a more fundamental model is linearized around a point of interest. While Hermitian matrices are common in canonical quantum mechanics, the richer behavior of non-Hermitian matrices is being increasingly used to model gain/loss in open systems~\cite{elganainy2018nature, yoshida2019ring}, phase transitions~\cite{melkani2021polymers,fruchart2021nonreciprocal,hamazaki2019localization}, sensitivity to boundary conditions~\cite{Borgnia2020boundary,lin2023skin}, and various other phenomena excluded by assumptions of Hermiticity.

Pseudo-Hermitian matrices~\cite{mostafazadeh2010pseudo} are non-Hermitian matrices that can be similarity-transformed to their adjoints, $H = G^{-1} H^\dagger G$. They are ubiquitous in classical physics~\cite{melkani2021polymers, yoshida2019ring, Susstrunk2016phonons,kruss2022amplification,shankar2017topoflock} since all real-valued matrices are pseudo-Hermitian. Matrices with time-reversal symmetry~\cite{kawabata2019timereversal} or with parity-time symmetry ($\mathcal{PT}$-symmetry)~\cite{ozdemir2019parity} are also pseudo-Hermitian (see \secref{sec:definitions})---the latter being one of the earliest classes of non-Hermitian matrices to be analyzed in terms of symmetries~\cite{bender2005PT, ashida2020nonhermitian}.

Upon tuning some parameter, degenerate real eigenvalues of a pseudo-Hermitian matrix can turn into complex conjugate pairs. This phenomenon is known as spontaneous pseudo-Hermiticity breaking (henceforth simply called symmetry breaking) since it is accompanied by a change in the symmetries of the corresponding eigenvectors. Typically, the parameter quantifies an external source of bias, amplification, or dissipation~\cite{melkani2021polymers, elganainy2018nature}. Upon symmetry breaking the system exhibits qualitatively different behavior usually signifying the emergence of amplified/dissipated modes~\cite{elganainy2018nature} or even different thermodynamic phases~\cite{fruchart2021nonreciprocal,hamazaki2019localization,wang2022realcomplex}. Systematically analyzing the conditions for pseudo-Hermiticity breaking is then crucial to understanding the physical properties and potential applications of pseudo-Hermitian systems.

In this paper, we provide the necessary and sufficient conditions for symmetry breaking to occur in a pseudo-Hermitian matrix, $H$. We use the intertwining operator $G$, whose expectation value is a conserved quantity~\cite{ruzsicka2021conserved}, to classify the eigenspace of $H$. We demonstrate that symmetry breaking occurs when and only when eigenvalues associated with opposite signs of the conserved quantity collide on the real axis. This allows one to predict which real-valued degeneracies of a pseudo-Hermitian matrix can lead to symmetry breaking, and hence are ``unstable degeneracies". 

By characterizing the degeneracies we also determine the sets of pseudo-Hermitian matrices with real eigenvalues that can be continuously connected to each other without ever encountering symmetry breaking. These in turn correspond to all the disconnected stable phases (regions in parameter space where all eigenvalues are real) of a physical system.

Non-Hermitian matrices exhibit two types of eigenvalue degeneracies---diabolic points (DPs), where the number of independent eigenvectors equals the number of times an eigenvalue is repeated, and exceptional points (EPs)~\cite{heiss2012exceptional, kawabata2019exceptional,hu2022exceptionaltopo}, where the matrix cannot be diagonalized and its eigenvectors fail to span the complete space. We find that these disconnected regions in parameter space, where all eigenvalues are real, are surrounded by exceptional surfaces, which comprise all the real-valued EPs of pseudo-Hermitian matrices. Exceptional surfaces that are boundaries to two different regions may meet, annihilating each other and giving rise to DPs. These exceptional surfaces, together with the diabolic points created by their intersections, comprise all points of pseudo-Hermiticity breaking.

Furthermore, degeneracies in Hamiltonian matrices are either accidental or caused by symmetries. In this paper, we also use the intertwining operator to derive the conditions for when degeneracies caused by external symmetries are susceptible to thresholdless pseudo-Hermiticity breaking (i.e., symmetry breaking at infinitesimal amounts of non-Hermiticity).

Much of our paper builds on the insight by Refs. \cite{zhang2020PT,langer2004krein,zhang2016Krein} that pseudo-Hermitian matrices can be mapped to $G-$Hamiltonian matrices, a class of matrices studied by Krein, Gel'fand, Lidskii, and others in the context of stability of mechanical systems~\cite{YakubovichStarzhinskii,coppel1965stability,chang2019bifurcation,mochizuki2021fate}. This mapping is being increasingly used in works on quadratic Bosonic systems~\cite{Flynn2020restoring, peano2018topological} that can be described by an effective Bogoliubov-de Gennes (BdG) Hamiltonian that is pseudo-Hermitian. (See Ref.~\cite{Flynn2020deconstructing} for a detailed analysis.)

This article is structured as follows. We define pseudo-Hermitian matrices and intertwining operators in \secref{sec:definitions} and review the phenomenon of pseudo-Hermiticity breaking. In \secref{sec:structure} we use the intertwining operator to classify the eigenspace. In \secref{sec:conditions} we provide the main results of the paper including the conditions for pseudo-Hermiticity breaking to occur and the classification of stable phases of a pseudo-Hermitian matrix. These results are illustrated via a schematic example in \secref{sec:schematic}. In \secref{sec:boundaries} we characterize the boundaries of the stable phases, i.e., the points of symmetry breaking. We then discuss conserved quantities in pseudo-Hermitian systems and the interplay of pseudo-Hermiticity with degeneracies caused by external symmetries in \secref{sec:symmetries}. In Sec.~\ref{sec:examples}, we provide illustrative examples of well-known pseudo-Hermitian Hamiltonians from photonics, condensed matter physics, and mechanics.

\section{Pseudo-Hermitian matrices and symmetry breaking}\label{sec:definitions}

A matrix $H$ is called pseudo-Hermitian if it is similar to its conjugatectranspose. That is
\begin{equation}\label{eq:defnPseudo}
H = G^{-1} H^\dagger G
\end{equation}
for some invertible matrix $G$ called the intertwining operator~\cite{mustafa2008intertwining,ruzsicka2021conserved}. $G$ is not unique and can always be chosen to be Hermitian~\cite{zhang2020PT}, which we will assume is the case from here on. Since every matrix is similar to its transpose, an equivalent definition is that $H$ is similar to its complex conjugate,
\begin{equation}\label{eq:defnPseudo2}
    H = S H^* S^{-1}
\end{equation}
for some invertible matrix $S$. A real matrix is then trivially pseudo-Hermitian. The equation above can also be written as $[H,S\mathcal{T}] = 0$ where $\mathcal{T}$ is the antilinear complex-conjugation operator, which acts as the time-reversal operator. For this reason, $H$ is also said to be $S\mathcal{T}$ symmetric---a familiar case is when $S$ is the parity operator $\mathcal{P}$.

Non-Hermitian matrices with the $\mathcal{K}$ or $\mathcal{Q}$ internal symmetries (in the Bernard-LeClair notation~\cite{liu2019topodefects,kawabata2019symmetry, zhou2019topology}) are all special cases of pseudo-Hermitian matrices. Matrices satisfying $H = -G^{-1} H^\dagger G$ with an additional minus sign (or indeed any phase factor) can be transformed, via $H\to iH$, to also satisfy \eqnref{eq:defnPseudo}. 

Pseudo-Hermiticity breaking occurs when a (degenerate) real eigenvalue splits into a complex conjugate pair on the variation of a parameter, such as the Bloch wave vector for a periodic system. If $\lambda$ is a complex-valued eigenvalue with associated eigenvector $\ket{R}$ then, by \eqnref{eq:defnPseudo2}, $S \ket{R}^* = S\mathcal{T}\ket{R}$ is an eigenvector of $H$ associated with $\lambda^*$ and is, thus, linearly independent of $\ket{R}$. Conversely, if $\lambda$ were real and nondegenerate, then $S \mathcal{T} \ket{R}$ and $\ket{R}$ would be linearly dependent, i.e., $\ket{R}$ would be an eigenvector of $S \mathcal{T}$. In general, eigenvectors of $H$ with real eigenvalues can be chosen to also be eigenvectors of the ``symmetry operator" $S \mathcal{T}$. On the variation of a parameter, when a degenerate real eigenvalue splits into complex conjugate pairs, this symmetry of the eigenvectors gets spontaneously broken: $S \mathcal{T} \ket{R}$ and $\ket{R}$ become linearly independent. 

While the above formulation in terms of $S$ and $\mathcal T$ is familiar~\cite{bender2005PT,ashida2020nonhermitian}, in the following we formulate this behavior in terms of $G$. This will enable us to uncover additional features including the conditions for symmetry breaking to occur.

\section{Structure of the eigenspace}\label{sec:structure}

We denote a column vector by the ket $\ket{v}$ and its conjugate transpose $(\ket{v})^\dagger$ by the bra $\bra{v}$. The right and left eigenvectors of $H$ are defined by $H\ket{R_i} = \lambda_i\ket{R_i}$ and $\bra{L_i}H = \lambda_i \bra{L_i}$. They share the same eigenvalues. Taking conjugate transpose of the latter equation we get
\begin{equation}
H^\dagger \ket{L_i} = \lambda_i^* \ket{L_i}.
\end{equation} 
Operating \eqnref{eq:defnPseudo} on $G^{-1}\ket{L_i}$ shows that $G^{-1}\ket{L_i}$ is a (right) eigenvector of $H$ with eigenvalue $\lambda_i^*$.

First let us assume that there are no degeneracies in the eigenvalues of $H$. In that case we can define a biorthonormal eigenbasis for $H$~\cite{ashida2020nonhermitian},
\begin{equation}\label{eq:Hdiagonal}
    H = \sum_i \lambda_i \ket{R_i} \bra{L_i}, \quad \expect{L_i|R_j} = \delta_{ij}.
\end{equation}
Now if $\ket{R_i}$ is the eigenvector of $H$ with $\operatorname{Im}\lambda_i\neq 0$ then the eigenvector corresponding to the eigenvalue $\lambda_j = \lambda_i^*$ is $\eta G^{-1}\ket{L_i}$ where $\eta$ is some constant. Furthermore,
\begin{equation}
    \expect{R_j|G|R_j} = \eta \expect{R_j|G G^{-1}|L_i} = 0.
\end{equation}

On the other hand if $\lambda_i$ is real we should have $\ket{R_i} = \mu G^{-1}\ket{L_i}$ ($\mu$ being some non-zero constant). In this case,
\begin{equation}
    \expect{R_i|G|R_i} =  \mu\expect{R_j|G G^{-1}|L_i} = \mu.
\end{equation}
Since $G$ is invertible and Hermitian, $\mu$ has to be a nonzero real constant.

On relaxing our assumption of no degeneracies, these statements generalize as follows (see \appref{app:miniProofs} for details):

If $\lambda$ is complex valued, $\operatorname{Im} \lambda \neq 0$, then $\expect{R|G|R} = 0$ for all associated eigenvectors $\ket{R}$.

If $\lambda$ is real valued, there are three possibilities. If $\expect{R|G|R}$ is positive for all vectors, $\ket{R}$, in the eigenspace, we say the eigenvalue $\lambda$ is of \emph{positive kind} (the first kind in Krein's formulation~\cite{YakubovichStarzhinskii}). Similarly, if $\expect{R|G|R}$ is always negative, the eigenvalue $\lambda$ is of \emph{negative kind}. The third possibility is that one is able to find two eigenvectors, $\ket{R_j}$ and $\ket{R_k}$ in the eigenspace, such that $\expect{R_j|G|R_j}$ and $\expect{R_k|G|R_k}$ are of opposite signs. In this case the eigenvalue is said to be of \emph{indefinite kind} and it is possible to find an eigenvector $\ket{R_m}$, say, such that $\expect{R_m|G|R_m} = 0$. Exceptional point degeneracies are always of indefinite kind while non-degenerate real eigenvalues are never indefinite.

We are now ready to state the main result of the paper, which connects symmetry breaking to the expectation value of the intertwining operator $G$. 

\section{Conditions for symmetry breaking}\label{sec:conditions}

Suppose $G(\bs k)$ is a (continuously) parameterized matrix that is Hermitian and invertible for all values of the parameter(s) $\bs k$; and suppose $H(\bs k)$ is a (continuously) parameterized pseudo-Hermitian matrix obeying $H(\bs k) = G(\bs k)^{-1} H(\bs k)^\dagger G(\bs k)$. Let $\bs k = \bs k_0 + \epsilon \bs q$ where $\epsilon$ is small, $\bs q $ is an arbitrary direction in parameter space, and $\bs k_0$ is some reference point such that $H(\bs k_0)$ has all its eigenvalues real. 

$H(\bs k_0)$ is said to be protected from pseudo-Hermiticity breaking (strongly stable in Krein's formulation) if for any $\bs q $, $H(\bs k = \bs k_0 + \epsilon \bs q)$ also has all real eigenvalues. For example, if $H(\bs k_0)$ has no degeneracies then it is protected from symmetry breaking. The necessary and sufficient conditions for symmetry breaking to occur in a pseudo-Hermitian matrix are provided by the Krein-Gel'fand-Lidskii (KGL) Theorem.

\textit{Krein-Gel'fand-Lidskii Theorem. } $H(\bs k_0)$ is protected from pseudo-Hermiticity breaking if and only if the eigenvalues of $H(\bs k_0)$ are definite, i.e., $\expect{v|G(\bs k_0)|v}\neq 0$ for all eigenvectors $\ket{v}$ of $H(\bs k_0)$. The proof of this statement can be found in \appref{app:KGL} and in Refs.~\cite{coppel1965stability,YakubovichStarzhinskii}.

In particular, this implies that while an exceptional point at a real eigenvalue is sufficient for symmetry breaking to occur (since $\expect{v|G|v} = 0$ at an EP), it is not necessary. (We provide explicit examples of pseudo-Hermitian matrices with symmetry breaking at diabolic points in the sections below.)

Eigenvalues of positive and negative kind retain their kind on the variation of a parameter (see \appref{app:continuity}). This provides predictive power for the purposes of engineering (or avoiding) exceptional points and symmetry-breaking points. For example, suppose that on the variation of a parameter, two definite real eigenvalues are about to meet on the real axis. If the eigenvalues are of the same kind ($\expect{v|G|v}$ is of the same sign) then even after colliding they are forbidden from moving off the real axis or giving rise to an exceptional point degeneracy.

The matrix $H$ has as many eigenvalues of positive (negative) kind as the number of positive (negative) eigenvalues of the intertwining operator $G$
(see \appref{app:continuity}). In particular, if $G$ were positive-definite (or negative-definite) then all eigenvalues of $H$ would always be real and all degeneracies diabolic. $H$ would then be equivalent to a Hermitian matrix, a condition known as \textit{exact} pseudo-Hermiticity~\cite{mostafazadeh2003exact} or quasi-Hermiticity~\cite{ashida2020nonhermitian}.

Suppose an $N\times N$ pseudo-Hermitian matrix, protected from symmetry breaking, has eigenvalues, $\lambda_1\leq \dots \leq \lambda_N$. We can characterize it by a signature---an ordered list of $N$ signs, such as $(+,+,-,+,\dots)$, where the $n$th sign signifies the kind of the $n$th eigenvalue. Two strongly stable pseudo-Hermitian matrices can be continuously connected to each other, without every encountering pseudo-Hermiticity breaking, if and only if they have the same signature~\cite{YakubovichStarzhinskii}. If $G$ has $p$ positive eigenvalues (and hence $N-p$ negative eigenvalues), then the number of possible distinct signatures is $\frac{N!}{p!(N-p)!}$ (or $N$ choose $p$). These characterize all the disconnected regions of parameter space where the eigenvalues of a pseudo-Hermitian matrix are all real. In many physical systems these correspond to distinct phases characterized by localization, absence of dissipation, etc. (see \secref{sec:examples}).

The matrices at the boundaries enclosing these regions protected from symmetry breaking (i.e., at the points at which symmetry breaking occurs) have eigenvalues of indefinite kind. In \secref{sec:boundaries} we show that if two strongly stable regions share a point on the boundary, the point is diabolic. All other points on the boundaries are exceptional degeneracies and can be uniquely associated with a single strongly stable region and its signature. It is known that parametrized pseudo-Hermitian matrices have EPs with codimension $1$ and DPs with codimension $3$~\cite{fruchart2021nonreciprocal,seyranian2005coupling,delplace2021multifold}. The large codimensionality of DPs compared to that of EPs explains why symmetry breaking is observed usually at EPs. 

It is useful to compare this characterization with the classification of symmetry-protected topological (SPT) phases in non-Hermitian matrices~\cite{Wojcik2020homotopy,liu2019topodefects,kawabata2019symmetry, zhou2019topology} (which generalizes the classification of Hermitian matrices~\cite{chiu2016reviewHermitian}). Matrices belonging to the same symmetry class/phase are topologically equivalent---they are characterized by the same topological invariants and they can be continuously deformed into each other without encountering any degeneracy (the spectral gap is always open). For such a classification all degeneracies are considered equivalent. Here, we are considering properties of the individual eigenspaces of a matrix that too remain invariant unless there are degeneracies and the eigenspaces merge. Moreover, we distinguish between two kinds of degeneracies on physical grounds---stable ones that prohibit symmetry breaking and unstable ones that allow symmetry breaking. 

The KGL theorem was originally formulated with the assumption that the intertwining operator $G$ is constant and does not depend on the parameter(s) $\bs k$~\cite{YakubovichStarzhinskii}. Many systems, especially real-valued systems, however, have intertwining operators that are continuous functions of the parameters. Our proof of the KGL theorem in \appref{app:KGL} and \appref{app:continuity} generalizes to the case where $G(\bs k)$ is a continuous function of the parameter(s) $\bs k$ as long as it is Hermitian and invertible for all values of $\bs k$.

Since the intertwining operator $G$ is not unique, we provide an exhaustive method to find all intertwining operators for a pseudo-Hermitian matrix $H$ in \appref{app:finding} and show that they form a vector space. Our results will apply to each Hermitian intertwining operator that one can find. Finally, since pseudo-Hermiticity is equivalent to commutation with a generalized $\mathcal{P}\mathcal{T}$ operator, in \appref{app:PT} we provide an attempt to formulate the results of this section in terms of $\mathcal{P}$ and $\mathcal{T}$.

\section{A schematic example}\label{sec:schematic}

Let us illustrate these results with an example. Suppose $H$ is a $3\times 3$ pseudo-Hermitian matrix with respect to the conveniently diagonalized,
\begin{equation}
    G = \begin{pmatrix}
    1 & 0 & 0\\
    0 & 1 & 0\\
    0& 0 & -1
    \end{pmatrix}.
\end{equation}

The general form of $H$ can be written as,
\begin{equation}
    H = \begin{pmatrix}
    r_1 & a & b\\
    a^* & r_2 & c\\
    -b^*& -c^* & r_3
    \end{pmatrix},
\end{equation}
where $r_i$ are real numbers and $a,b,c$ are complex numbers. When $a=b=c=0$, the three eigenvalues are real and located at $\lambda_i = r_i$. The eigenvectors are the permutations of $(1,0,0)^T$ and it is straightforward to check that $\lambda_3$ is of opposite kind to the other two eigenvalues. 

On increasing $|a|$, $\lambda_3$ remains fixed while $\lambda_1$ and $\lambda_2$ move away from each other on the real axis due to level repulsion. (The matrix is Hermitian when $b=c=0$.) If $|c|$ were increased with $a=b=0$, then $\lambda_2$ and $\lambda_3$ move, first towards each other along the real axis and then since they are of opposite kinds, away from each other in the complex plane after colliding. Similarly, if $|b|$ were increased then $\lambda_1$ and $\lambda_3$ move towards each other.

\begin{figure}[htb]
	\includegraphics[width=\columnwidth]{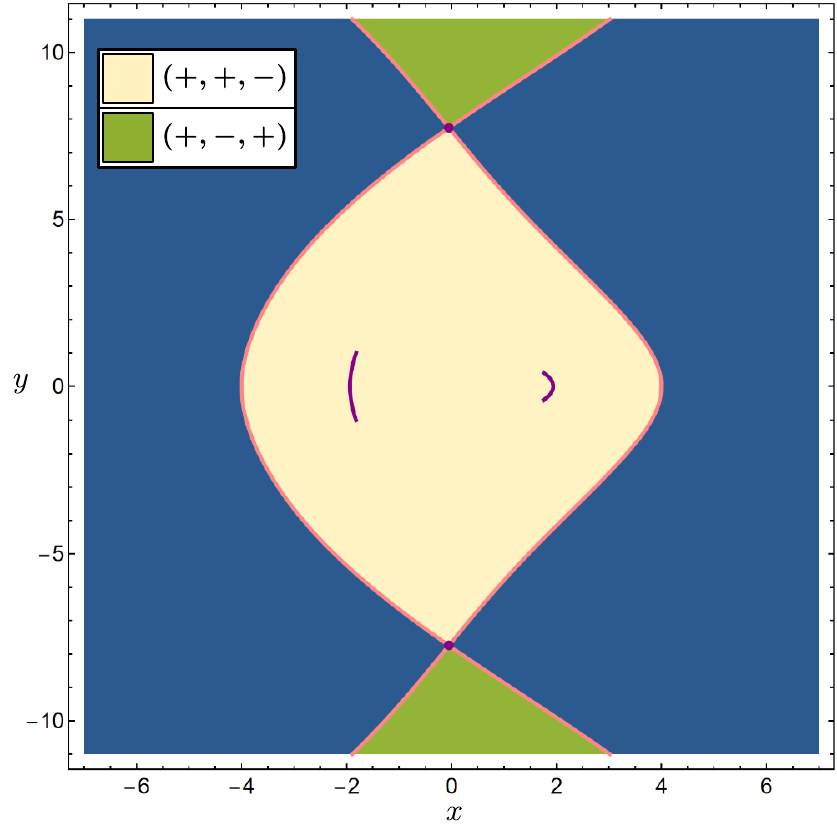}
	\caption[]{\label{fig:regions} 
	Parameter space of the matrix $H(x,y)$ in \eqnref{eq:schematic} which obeys $H = G^{-1} H^\dagger G$ where $G = \operatorname{diag}(1,1,-1)$. In the blue region of parameter space, pseudo-Hermiticity is broken and the eigenvalues of $H$ are no longer all real. The yellow region is protected from symmetry breaking and has signature $(+,+,-)$, i.e., the eigenvalue of negative kind is the largest. Similarly, the green region has signature $(+,-,+)$. Regions of different signatures are topologically disconnected. The boundaries between these regions have exceptional point degeneracies (pink curves), except where the curves meet and annihilate to form a diabolic point (purple points). The central region also contains two disconnected diabolic curves (purple curves) involving two eigenvalues of positive kind. Since they are of the same kind these degeneracies cannot lead to symmetry breaking.
	}
\end{figure}

Let us constrain $H$ to a two-dimensional parameter space with the parametrization,
\begin{equation}\label{eq:schematic}
    H(x,y) = \begin{pmatrix}
    1 & y & ix\\
    y & \frac{3}{2} & \frac{i}{2}\sin(\frac{\pi y}{8})\\
    ix& \frac{i}{2}\sin(\frac{\pi y}{8}) & 9
    \end{pmatrix}.
\end{equation}
In \figref{fig:regions}, we show the parameter space of this matrix. The region around the origin, $(0,0)$, is protected from symmetry breaking and has the signature $(+,+,-)$. As we move from the origin along the positive $x$ axis, the eigenvalue at $r_2 = \frac{3}{2}$ remains fixed while the other two eigenvalues move towards each other. At $x=\frac{\sqrt{15}}{2} \approx 1.94$ a diabolic degeneracy is encountered but since it involves eigenvalues of the same kind, it cannot give rise to symmetry breaking even in a higher dimensional parameter space. On increasing $x$ further, two eigenvalues of opposite kind meet at an EP at $x=4$ beyond which pseudo-Hermiticity gets broken.

Figure~\ref{fig:regions} shows that the boundary between symmetry broken regions of the $2$D parameter space and regions protected from symmetry breaking are given by exceptional curves. When the exceptional curves meet, they annihilate each other to form diabolic points~\cite{delplace2021multifold}. These are examples of symmetry breaking at DPs---perturbing the value of $x$ at the DP at $(x,y)\approx(-0.05,7.75)$, for example, will cause the eigenvalues to become complex valued. While it seems from the figure that the region with signature $(+,-,+)$ consists of disconnected areas, this is just an artifice of being constrained to a two-dimensional parameter space. 

\section{Characterizing the points of symmetry breaking}\label{sec:boundaries}

The points of symmetry breaking of an $N \times N$ pseudo-Hermitian matrix form the boundaries separating strongly stable regions of parameter space from the regions where at least one eigenvalue is complex valued. To characterize these points we will only focus on the most relevant case of boundaries formed when two eigenvalues of opposite kind meet each other. (Higher order degeneracies are rare in the absence of additional symmetries.) In the subspace corresponding to these two eigenvalues of opposite kind, a general intertwining operator (after diagonalization) is given by
\begin{equation}\label{eq:originalG}
    G_0 = \begin{pmatrix}
        \eta_1 & 0\\
        0 & -\eta_2
    \end{pmatrix}
\end{equation}
where $\eta_1$ and $\eta_2$ are positive real numbers. The relation $H_0 = G_0^{-1} H_0^\dagger G_0$ eliminates $4$ out of $8$ real parameters of a generic complex-valued $2\times 2$ matrix. We can parametrize such a pseudo-Hermitian matrix as
\begin{equation}\label{eq:originalH}
    H_0 = \begin{pmatrix}
        \lambda + a & \eta_2 b e^{i\theta}\\
        -\eta_1 b e^{-i\theta} & \lambda - a
    \end{pmatrix}
\end{equation}
where $a,b,\theta,\lambda$ are real. Since $\theta$ does not enter the characteristic equation and $\lambda$ only shifts the eigenvalues by the same amount, the physics of a generic $2\times 2$ pseudo-Hermitian matrix is controlled by the gain-loss factor $(a)$ and the coupling factor $(b)$ (see \secref{sec:exa:qubit}). The matrix has degenerate eigenvalues (at $\lambda$) when we set $b = \frac{a}{\sqrt{\eta_1 \eta_2}}$,
\begin{equation}
    H_0 = a\begin{pmatrix}
        1 & \sqrt{\frac{\eta_2}{\eta_1}} e^{i\theta}\\
        -\sqrt{\frac{\eta_1}{\eta_2}} e^{-i\theta} & -1
    \end{pmatrix} + \begin{pmatrix}
        \lambda & 0\\
        0 & \lambda
    \end{pmatrix}.
\end{equation}
The behavior of the degeneracy is controlled by the parameter $a$ --- the degeneracy is diabolic when $a = 0$ and is exceptional at all other values. 

In \appref{app:boundaries}, we consider pseudo-Hermitian matrices arbitrarily close to $H_0$. That is, we consider $G(\epsilon), H(\epsilon)$ such that $H(\epsilon) = G(\epsilon)^{-1} H(\epsilon)^\dagger G(\epsilon)$ and $G(0) = G_0, H(0)=H_0$. In the limit of $0<\epsilon \ll <1$, we find that strongly stable matrices close to $H_0$ with $a>0$ can only be from the region $(-,+)$ while those close to $H_0$ with $a<0$ can only be from the region $(+,-)$. In contrast, one can always find strongly stable matrices from both the regions $(-,+)$ as well as $(+,-)$ arbitrarily close to $H_0$ with $a=0$ (the DP) (see \figref{fig:twoByTwo}).

\begin{figure}[htb]
	\includegraphics{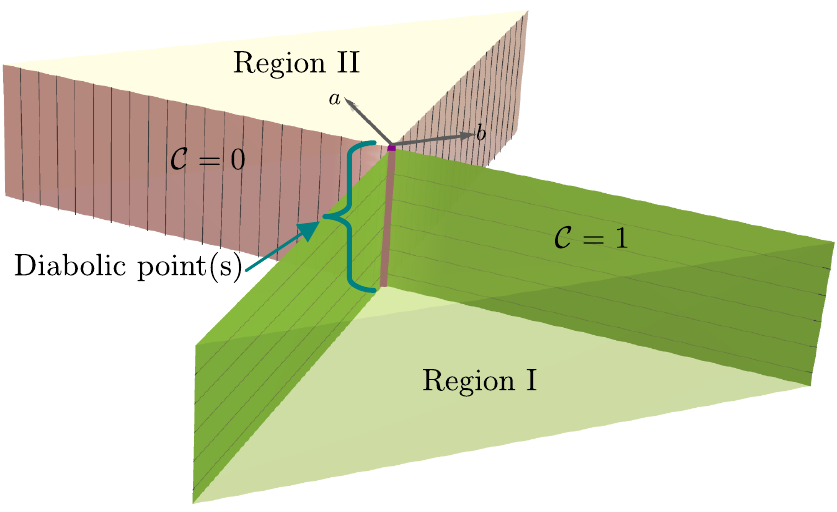}
	\caption[]{\label{fig:twoByTwo} 
	Strongly stable regions in parameter space (where all eigenvalues are real) are surrounded by exceptional surfaces that meet and annihilate at diabolic points (purple line). These EPs and DPs together make up all possible points where pseudo-Hermiticity breaking occurs. Two such strongly stable regions with signature $(\dots, +, -, \dots)$ and $(\dots, -, +, \dots)$ are shown (labelled I and II respectively). The exceptional surfaces (green with horizontal mesh and pink with vertical mesh respectively) can be identified uniquely with the strongly stable region they enclose and the corresponding signature. Their topological charges can also be written in terms of the zeroth Chern number ($\mathcal{C}=1$ and $\mathcal C = 0$ respectively). Pseudo-Hermiticity is broken in the region outside these surfaces. In the neighbourhood of a symmetry-breaking point, the Hamiltonian can be reduced to the two parameter matrix $\begin{pmatrix} a & b\\ -b & -a \end{pmatrix}$ where $a\in \mathbb R$ is a gain-loss parameter and $b \in \mathbb R$ is a coupling parameter [see \eqnref{eq:originalH}].
	}
\end{figure}

Thus, the EPs on the boundaries can be identified with the topological index, i.e., the signature, of the strongly stable region they enclose. In fact, EPs can be characterized by several independent topological indices~\cite{kawabata2019exceptional}. For example, they can also be characterized by the topological indices associated with symmetry protected exceptional surfaces/rings that arise in systems with chiral symmetry~\cite{yoshida2019ring, yoshida2020exceptional}. This relies on the fact that $iH$ has the non-Hermitian chiral symmetry, $(iH) G + G (iH)^\dagger = 0$ and thus one can follow the procedure first laid out in Ref.~\cite{yoshida2019ring}. The relevant topological index is the zeroth Chern number number associated with an extended Hermitian matrix created from $H$ (see \appref{app:topological}). We find that the exceptional line $a>0$ carries a Chern number of $0$ while the line $a<0$ has Chern number $1$ (see \figref{fig:twoByTwo}).

\section{Relation with symmetries}\label{sec:symmetries}

Degeneracies in Hamiltonian matrices are either accidental or caused by symmetries. In Hermitian Hamiltonians, the latter case is more common due to level repulsion. It is natural to ask whether degeneracies caused by symmetries are susceptible to pseudo-Hermiticity breaking? Indeed if this were the case we would observe \emph{thresholdless} pseudo-Hermiticity breaking --- even a small non-Hermitian perturbation of a spatially symmetric Hermitian Hamiltonian would cause complex-valued eigenvalues (see Ref.~\cite{ge2014degeneracy} and \secref{sec:exa:lattice} for examples). We show that the intertwining operator provides the right language to address this question.

We recall that $H$ has a symmetry described by a group $\mathcal{G}$ if $H$ commutes with the matrix representations of the elements of $\mathcal{G}$. We assume that the representations are unitary---a standard assumption in canonical quantum mechanics since it keeps the inner product invariant. Additionally, we recall that any representation $\eta$ can be generically broken into irreps (irreducible representations), $\eta = \pi_1\oplus \pi_2\oplus \dots$, and that the Hamiltonian $H$ becomes degenerate in each subspace on which an irrep acts~\cite{Woit2017quantum}. This is the reason why symmetries typically give rise to degeneracies unless, for example, the underlying group is Abelian, in which case the irreps are one-dimensional.

Let us start below the threshold for symmetry breaking such that a pseudo-Hermitian matrix $H$ has all real eigenvalues and can be diagonalized. Now on tuning a parameter suppose the spatial symmetry of the Hamiltonian is explicitly broken, which allows for the previously degenerate eigenvalues to move away from each other. Would they stay on the real axis or would they move off in the complex plane?

To answer this question we need to classify the subspace $\Omega$ on which an irreducible representation $\pi$ acts. Before spatial symmetry is explicitly broken, $\Omega$ was also one of the degenerate eigenspaces of $H$. Now, since $\pi$ is an irrep, for any non-zero vector $\ket{v}\in \Omega$ one can operate the group elements, $\pi(g_i)\pi(g_j)\dots\ket{v}$ to span the whole space $\Omega$. Computing the expectation value of the intertwining operator $G$ for each of these spanning vectors we get terms like
\begin{equation}
    \bra{v}\dots \pi(g_j)^\dagger\pi(g_i)^\dagger G\pi(g_i)\pi(g_j)\dots\ket{v}.
\end{equation}
The relative sign of these terms capture whether or not we would see thresholdless symmetry breaking. We see that if $G$ were to commute with the elements of $\mathcal{G}$, all these terms would equal $\bra{v} G \ket{v}$ implying that the subspace is either of positive or negative kind (as long as $\bra{v} G \ket{v}$ is non-zero). The commutation of the intertwining operator, $G$, with the group elements thus ensures that the degeneracy is stable such that on loss of spatial symmetry even though the degeneracy of eigenvalues is broken, they still stay on the real axis.

We note here that the intertwining operator $G$ has another property associated with symmetries ---it provides a conserved quantity for pseudo-Hermitian systems~\cite{ruzsicka2021conserved}. To show this, we first define the time propagation matrix $U(t) = e^{-i H t}$, which is the solution to the Schr\"{o}dinger equation, $i\frac{dU}{dt} = HU$. Since $H$ is not Hermitian, $U(t)$ is not unitary; instead it satisfies~\cite{liu2022floquet}
\begin{equation}\label{eq:eqnU}
U^{-1}(t) = G^{-1}U^\dagger(t) G.
\end{equation}
The expectation value of $G$ with respect to a time-evolving vector $\ket{v(t)} = U(t)\ket{v(0)}$ is independent of time.
\begin{align}
\expect{v(t)|G|v(t)} &= \expect{v(0)|U^\dagger(t) G U(t)|v(0)} \nonumber\\
&= \expect{v(0)| G |v(0)}.
\end{align}
$\expect{v(t)|G|v(t)}$ is then a conserved quantity of the system ~\cite{ruzsicka2021conserved}. 

\section{Examples}\label{sec:examples}

In the following, we provide some examples of physical Hamiltonians exhibiting pseudo-Hermiticity breaking. We study a minimal example of a two-level system from photonics in \secref{sec:exa:qubit}, in particular noting the role of the conserved quantity for a system with gain and loss. In \secref{sec:exa:lattice}, we study a lattice in which non-Hermiticity arises from asymmetric couplings and show how pseudo-Hermiticity breaking can be analyzed in the submatrices of the Hamiltonian. Finally, in \secref{sec:exa:oscillators}, we examine coupled dissipative oscillators and find an interplay of two simultaneous intertwining operators.

These examples demonstrate how various physical systems that were previously considered to exhibit unique rich behavior can be analyzed from the unified perspective of pseudo-Hermiticity.

\subsection{Gain and loss in a qubit}\label{sec:exa:qubit}

One of the simplest pseudo-Hermitian system is a two-level system describing the physics of two coupled sites, one experiencing gain and the other suffering loss.
The Hamiltonian
\begin{equation}
    H = \begin{pmatrix}
    -ig_1 & \kappa\\
    \kappa & ig_2
    \end{pmatrix},
\end{equation}
where $g_1,g_2$ describe the amplification/dissipation at each site and $\kappa$ is the coupling constant describes such a system. See the review articles Refs.~\cite{ozdemir2019parity, elganainy2018nature} for examples of physical setups described by such a Hamiltonian. When $g_1=g_2=g$, the gain and loss are balanced and $H$ becomes pseudo-Hermitian, $H = G^{-1} H^\dagger G$, with
\begin{equation}
    G= \begin{pmatrix}
    0 & 1\\
    1 & 0
    \end{pmatrix}.
\end{equation}

If $\begin{pmatrix} a(t)\\ b(t) \end{pmatrix} 
= \begin{pmatrix} |a(t)|e^{i\theta_1(t)}\\ |b(t)|e^{i\theta_2(t)} \end{pmatrix}$ 
is a solution to the Schr\"{o}dinger equation, then the conserved quantity is
\begin{equation}
    C = \begin{pmatrix}
    a^* & b^*
    \end{pmatrix}\begin{pmatrix}
    0 & 1\\
    1 & 0
    \end{pmatrix}\begin{pmatrix}
    a\\
    b
    \end{pmatrix} = 2|a(t) b(t)|\cos(\theta_1(t) - \theta_2(t)).
\end{equation}

Furthermore, $G$ has eigenvalues $+1$ and $-1$ so the two modes of $H$ are of opposite kinds and we can expect symmetry breaking. (Indeed for a two-level system we either have this case or the trivial case of both modes being of the same kind for which the dynamics is similar to Hermitian dynamics.)

\begin{figure}[htb]
	\includegraphics{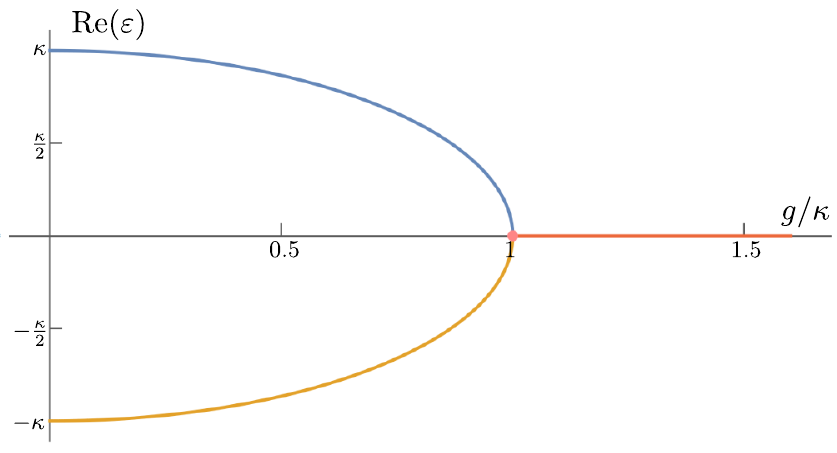}
	\caption[]{\label{fig:qubit} 
	Real part of the eigenvalues of $H = \begin{pmatrix} -ig & \kappa\\ \kappa & +ig \end{pmatrix}$ which describes two coupled modes with equal and opposite gain/loss. When $g< \kappa$, the two eigenvalues, which are of opposite kind, are real. They meet each other at an exceptional point when $g/\kappa$ reaches the threshold value of $1$ and on further increasing $g/\kappa$, the eigenvalues become complex valued.
	}
\end{figure}

The eigenvalues of $H$ are $\pm \sqrt{\kappa^2 - g^2}$. When $\kappa>g$, the coupling dominates the amplification/dissipation, the eigenvalues are real, and pseudo-Hermiticity symmetry is unbroken. Symmetry breaking occurs when $g$ reaches its critical value, $g_c = \kappa$, after which the eigenvalues of $H$ become complex (see \figref{fig:qubit}).

\textit{Symmetry unbroken phase ($g<\kappa$). }The eigenvalues are real and can be written as $\pm \eta$, where $\eta =\sqrt{\kappa^2 - g^2}$ is real and positive. The eigenvectors are $\ket{v_\pm} = \begin{pmatrix} 1\\ \pm e^{\pm i\theta} \end{pmatrix}$ where $\theta= \arcsin (g/\kappa)$. We have $\expect{v_s|G|v_s} = 2s\cos(\theta)$ so the eigenvalue $+\eta$ is of positive kind while $-\eta$ is of negative kind.

The general solution to the Schr\"{o}dinger equation is
\begin{equation}
    \ket{v(t)} = \begin{pmatrix} 
    c_1 \cos(\theta - \eta t) + i c_2 \sin(\eta t)\\ 
    c_2 \cos(\theta + \eta t) + i c_1 \sin(\eta t)
    \end{pmatrix}
\end{equation}
where $c_1,c_2$ specify the initial conditions. If we start with a $\frac \pi 2$ phase shift, $c_1 = e^{i\alpha}|c_1|,c_2 = ie^{i\alpha}|c_2|$, then the two sites remain phase locked due to the conservation of $\expect{v|G|v}$.

\textit{Symmetry broken phase ($g>\kappa$). }The eigenvalues are now complex conjugates and can be written as $\pm i\mu$, where $\mu =\sqrt{g^2 - \kappa^2}$ is real and positive. The eigenvectors of $H$ are $\ket{v_\pm} = \begin{pmatrix} \kappa \\ i(g \pm \mu) \end{pmatrix}$ and $\expect{v_s|G|v_s} = 0$.

The general solution to the Schr\"{o}dinger equation is
\begin{equation}
    \ket{v(t)} = \begin{pmatrix} 
    c_1\kappa e^{\mu t} + c_2\kappa e^{-\mu t}\\ 
    i c_1 (g+\mu) e^{\mu t} + i c_2 (g - \mu) e^{-\mu t}
    \end{pmatrix}.
\end{equation}
For nonzero $c_1$, since the amplitudes of both the sites increase exponentially for large $t$, asymptotically the sites develop a phase difference of $\frac \pi 2$. Again this is predicted directly from the conservation of $\expect{v|G|v}$ since as $|v_1 v_2|$ increases, $\cos(\theta_1 - \theta_2)$ should diminish.

\textit{Symmetry breaking point ($g=\kappa$). }At exactly $g=\kappa$, the system exhibits an exceptional point. The eigenvalues equal zero, and the eigenvectors coalesce to a single eigenvector, $\ket{v} = \begin{pmatrix} 1 \\ i \end{pmatrix}$ with $\expect{v|G|v} = 0$ as expected from an exceptional point.

While this two-level Hamiltonian is well known, the fact that it has a conserved quantity, which leads to a $\frac \pi 2$ phase-locking, was never noted. Furthermore, the physics of many finite-dimensional systems experiencing gain and loss cannot be captured in terms of effective two-level Hamiltonians~\cite{xiong2021higher, xiong2022blue}. These systems exhibit richer behavior with phase diagrams of pseudo-Hermiticity broken regions separated by exceptional lines and higher order exceptional points~\cite{xiong2021higher, xiong2022blue}. For these systems exact solutions are difficult, and using the intertwining operator to predict the phase diagram and the stability of degeneracies may be the only tractable method to lead to any physical insights.

\subsection{Lattice with asymmetric hopping}\label{sec:exa:lattice}

\begin{figure}[htb]
	\includegraphics{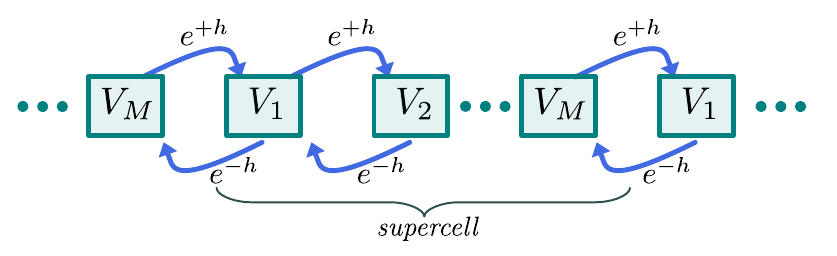}
	\caption[]{\label{Fig:latticeSchematic} 
	A lattice with asymmetric hopping described by the Hamiltonian in~\eqnref{eq:latticeHamiltonian}. Each supercell consists of $M$ sites with onsite potential $V_\alpha$. The particle hops to the site on its right with amplitude $e^{+h}$ and to the site on its left with amplitude $e^{-h}$.
	}
\end{figure}

We consider a one-dimensional real lattice of $N$ supercells, each with $M$ sites, connected by nearest-neighbor coupling (see \figref{Fig:latticeSchematic}). The Hamiltonian is~\cite{hebert2011hatanoNelson}
\begin{equation}\label{eq:latticeHamiltonian}
    H = \sum_{n=1}^N H_n
\end{equation}
where,
\begin{align}
    H_n &=\sum_{\alpha=1}^M V_\alpha \ket{n,\alpha}\bra{n,\alpha}\nonumber\\ 
    &+ \sum_{\alpha=1}^{M-1}\left ( e^{-h}\ket{n,\alpha}\bra{n,\alpha+1} + e^{h}\ket{n,\alpha+1}\bra{n,\alpha}\right)\nonumber\\ 
    &+ e^{-h}\ket{n,M}\bra{n+1,1} + e^{h}\ket{n+1,1}\bra{n,M}.
\end{align}
Here, $V_\alpha$ is the potential energy at each site, $e^h$ is the amplitude for the particle to hop rightwards, and $e^{-h}$ the amplitude to hop leftwards. Periodic boundary conditions are assumed.

Such a Hamiltonian has been used to model vortex lines in type II superconductors~\cite{hatano1996localizationprl} (where $h>0$ signifies a transverse component of the applied magnetic field), polymers chains on periodic substrates~\cite{melkani2021polymers} (where $h>0$ signifies an externally applied shear force), and a variety of other physical systems~\cite{li2021quantized,nelson198biology}. The source of non-Hermiticity in these models is asymmetric coupling instead of gain/loss as in the previous example. 

In the Hermitian limit, $h\to 0$, the eigenvalues of $H$ are real and form $M$ bands. As $h$ increases, it is known that the bands expand into ovals in the complex plane [see \subfigref{fig:latticeBands}{a}]~\cite{hebert2011hatanoNelson,melkani2021polymers}. As the bands expand, two bands may meet each other at a critical value of $h$ closing the bandgap and leading to an insulator-conductor transition. This corresponds to a localization-delocalization transition in the physical system~\cite{hatano1996localizationprl,melkani2021polymers,nelson198biology} and we show below that such a transition is equivalent to pseudo-Hermiticity breaking.

Due to the discrete translation symmetry of the lattice, $H$ can be block diagonalized as $H = \oplus_k \mathcal{H}(k)$ where 
\begin{align}
    \mathcal{H}(k) &=\sum_{\alpha=1}^M V_\alpha \ket{k,\alpha}\bra{k,\alpha}\nonumber\\ 
    &+ \sum_{\alpha=1}^{M-1} \left( e^{-h}\ket{k,\alpha}\bra{k,\alpha+1} + e^{h}\ket{k,\alpha+1}\bra{k,\alpha} \right )\nonumber\\ 
    &+ e^{-h-ik}\ket{k,M}\bra{k,1} + e^{h+ik}\ket{k,1}\bra{k,M},
\end{align}
and $k = \frac{\pi m}{N}$ with $m\in \{-N+2,-N+4,\dots,N-2,N\}$. In matrix form this is (where blanks denote zeros)
\begin{equation}\label{eq:latticeK}
    \mathcal H (k) = \begin{pmatrix}
    V_1 & e^{-h} &  & & & e^{h+ik}\\
    e^{h} & V_2 & e^{-h} &  & & \\
     & e^{h} & V_3 & e^{-h} &  &\\
     & &\ddots&\ddots&\ddots&\\
     & &  & e^h &V_{M-1} & e^{-h}\\
    e^{-h-ik} &  & &  & e^h &V_M
    \end{pmatrix}.
\end{equation}
The eigenvalues $\varepsilon(k)$ of $\mathcal{H}(k)$ correspond to the projections of the $M$ Bloch bands at wave-vector $k$. We will assume that the potential is inversion symmetric, such that $V_1 = V_M, V_2 = V_{M-1,\dots}$, which will simplify the form of the intertwining operator, and restrict ourselves to a lattice that is a band insulator at $h\to 0$ (all bands are separated). For concreteness we will also assume that $M$ is an odd number.

\begin{figure}[htb]
	\includegraphics{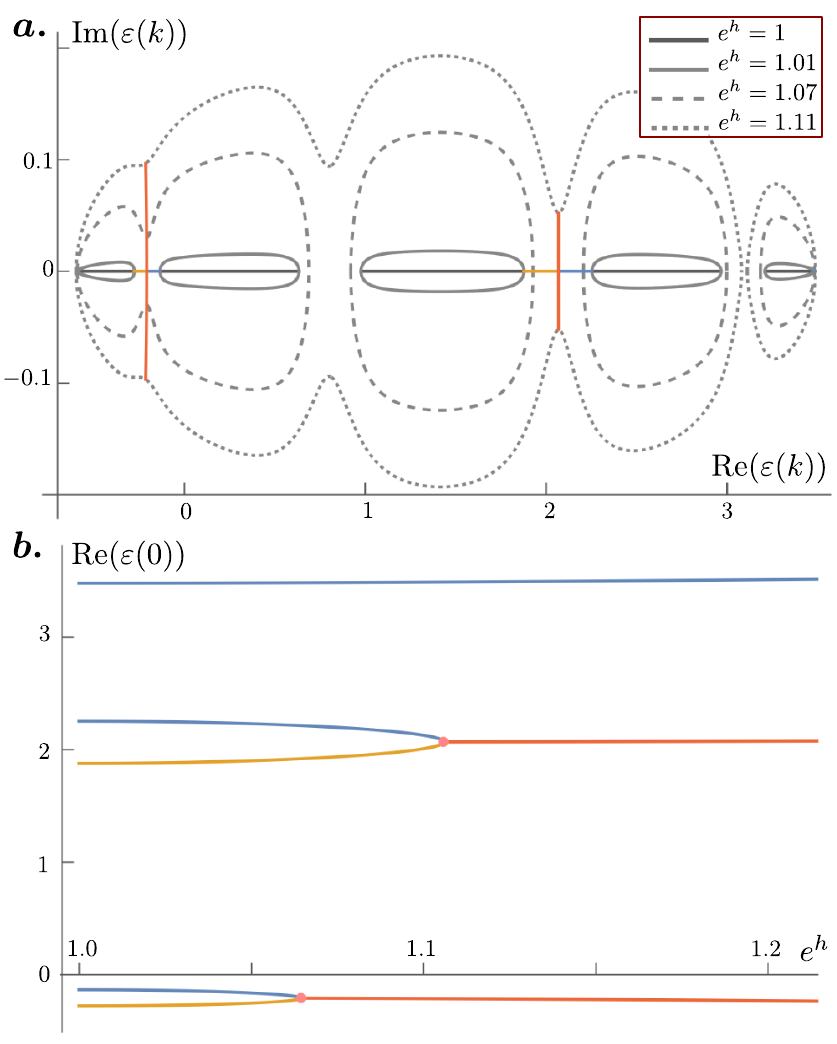}
	\caption{\subfiglabel{a} The eigenvalues of the Hamiltonian in \eqnref{eq:latticeHamiltonian} form M separate oval-shaped bands in the complex plane. Here, we have set $M=5$ and $V_1 = V_5 = 1.4,V_2 = V_3 = 1.2,V_3= 2$. At $h=0$ ($e^h = 1$), the Hamiltonian is Hermitian and the ovals collapse to line segments on the real axis [dark gray continuous lines in \subfiglabel{a}]. As the non-Hermiticity factor $h$ increases, the ovals grow in size and merge with each other (we show the bands at $e^h=1.01, 1.07, 1.11$). The eigenvalues associated with $k\neq 0,\pi$ move off the real axis at arbitrarily small $h$ while those with $k=0$ or $k=\pi$ move along the real axis until a degeneracy of indefinite kind is formed. The trajectories of the eigenvalues of $\mathcal{H}(0)$ [\eqnref{eq:latticeK}] on increasing $h$ are depicted by the continuous curves (real eigenvalues of positive and negative kind are shown in blue and yellow respectively, complex-valued eigenvalues are shown in red). These are also shown in \subfiglabel{b}  where the attraction between eigenvalues of opposite kinds is more apparent. The two symmetry-breaking points shown in pink are both exceptional degeneracies.
	}\label{fig:latticeBands}
\end{figure}

We note that $\mathcal{H}(k)=\mathcal{H}(-k)^*$ (due to the lattice being real valued) and $\mathcal{H}(k) = G^{-1}\mathcal{H}(k)^T G$ (due to the left-right symmetry of the lattice). Here, $G$ is an $M\times M$ antidiagonal matrix with each entry on the antidiagonal being $1$, i.e., 
\begin{equation}
    G = \begin{pmatrix}
    0 & 0 & \dots & 0 & 1\\
    0 & 0 & \dots & 1 & 0\\
    \vdots & \vdots & \iddots & \vdots & \vdots\\
    0 & 1 & \dots & 0 & 0\\
    1 & 0 & \dots & 0 & 0\\
    \end{pmatrix}.
\end{equation}
$G$ has $+1$ as an eigenvalue, repeated $\frac{M+1}{2}$ times, with any symmetric vector, $\{a_1,a_2,\dots,a_2,a_1\}^T$ as an eigenvector. The other eigenvalue is $-1$, repeated $\frac{M-1}{2}$ times, with antisymmetric vectors, $\{a_1,a_2,\dots, a_{\frac{M-1}{2}},0,-a_{\frac{M-1}{2}},\dots,-a_2,-a_1\}^T$ spanning the eigenspace. We use the pseudo-Hermiticity structure arising from these two symmetries to explain the key characteristics of this system.

The swelling of bands into ovals is captured by the matrix $\mathcal{H}(+k)\oplus \mathcal{H}(-k)$ where $k$ is neither $0$ nor $\pi$. For these matrices, one can show that $\Tilde G [\mathcal{H}(+k)\oplus \mathcal{H}(-k)] = [\mathcal{H}(+k)\oplus \mathcal{H}(-k)]^\dagger \Tilde G $ where $\Tilde G$ is a $2M\times 2M$ antidiagonal matrix with each entry on the antidiagonal being $1$. That is,
\begin{equation}
\begin{pmatrix}
    \bigz & G\\
    G & \bigz
\end{pmatrix}
\begin{pmatrix}
    \mathcal{H}(+k) & \bigz\\
    \bigz & \mathcal{H}(-k)
\end{pmatrix}
=
\begin{pmatrix}
    \mathcal{H}(+k) & \bigz\\
    \bigz & \mathcal{H}(-k)
\end{pmatrix}^\dagger
\begin{pmatrix}
    \bigz & G\\
    G & \bigz
\end{pmatrix}.
\end{equation}
$\Tilde G$ has $+1$ as an $M$-fold eigenvalue and $-1$ as the other $M$-fold eigenvalue.

In the Hermitian limit, $h\to 0$, the eigenvalues of $\mathcal{H}(k)$ are real for all values of $k$. The time-reversal symmetry, $\mathcal{H}(+k)=\mathcal{H}(-k)^*$, implies (by Kramers' theorem~\cite{zhou2019topology}) that the eigenvalues of $\mathcal{H}(+k)\oplus \mathcal{H}(-k)$ come in degenerate pairs. At nonzero $h$, Kramers' degeneracy theorem breaks down~\cite{zhang2022kramers} and we may expect a case of \emph{threshold-less} symmetry breaking~\cite{ge2014degeneracy} as in \secref{sec:symmetries} if the degenerate eigenvalues were of opposite kind.

Returning to the Hermitian limit, suppose $\ket{v_i(+k)}$ and $\ket{v_i(-k)}$ are the eigenvectors of $\mathcal{H}(+k)$ and $\mathcal{H}(-k)$, respectively, for a common eigenvalue $\lambda_i$. The superposition $\alpha\ket{v_i(+k)} + \beta\ket{v_i(-k)}$ is then the general eigenvector of $[\mathcal{H}(+k)\oplus \mathcal{H}(-k)]$ corresponding to $\lambda_i$. One can check that the eigenvector with $(\alpha,\beta) = (1,1)$ is of opposite kind to the one with $(\alpha,\beta) = (1,-1)$. Thus, all the degeneracies of $[\mathcal{H}(+k)\oplus \mathcal{H}(-k)]$ are of indefinite kind at $h\to 0$ enabling the eigenvalues to move off the real axis when non-Hermiticity is introduced at arbitrarily small $h$ leading to the swelling of bands into ovals [see \subfigref{fig:latticeBands}{a}].

Now, for $k=0,\pi$, we get $\mathcal{H}(0) = G^{-1}\mathcal{H}(0)^\dagger G$ and $\mathcal{H}(\pi)= G^{-1}\mathcal{H}(\pi)^\dagger G$. Thus, these two matrices have $\frac{M+1}{2}$ eigenvalues of positive kind and the rest of negative kind. In the Hermitian limit, $h\to 0$, the eigenvalues are all real and since $\mathcal{H}(0)$ [as well as $\mathcal{H}(\pi)$] commutes with $G$, its eigenvectors are either symmetric vectors (and of positive kind), or antisymmetric vectors (and of negative kind). Upon turning on the bias, by increasing the value of $h$, while the eigenvectors are no longer constrained to be either symmetric or anti-symmetric, the eigenvalues remain restricted to the real axis in the absence of degeneracies. Numerical investigations suggest that on increasing $h$, the eigenvalues of opposite kinds attract each other and when they meet, and two bands merge at $k=0$, they generically produce exceptional point degeneracies [see \subfigref{fig:latticeBands}{b} where we show the eigenvalues of $\mathcal{H}(0)$].

Whether or not the bands have merged is then captured by whether or not pseudo-Hermiticity breaking has occurred in the matrix $\mathcal{H}(0)\oplus \mathcal{H}(\pi)$. When the symmetry breaking does occur, it leads to a localization to delocalization transition (or insulator to metal transition) in the physical system.

In summary, we showed that the two key characteristics of a lattice with asymmetric coupling---namely, the swelling of bands into ovals at infinitesimal asymmetry and the localization-delocalization transition at a critical value of asymmetry---are captured by pseudo-Hermiticity breaking in the relevant submatrices of the system. We expect these key insights to be valuable in the study of higher dimensional systems~\cite{lehrer1998vortex}, dimerized systems, and systems with disorder~\cite{hebert2011hatanoNelson}.

\subsection{Coupled dissipative oscillators}\label{sec:exa:oscillators}

Harmonic oscillators are ubiquitous in classical physics since they model small fluctuations of a many-body system about its equilibrium configuration. In recent years, a detailed study of the matrix structure of mechanical oscillators revealed rich behavior such as internal symmetries~\cite{Susstrunk2016phonons}, topologically protected boundary modes~\cite{kane2014topological}, exceptional rings~\cite{yoshida2019ring}, etc. Here we consider a system of identical masses subject to an arbitrary harmonic potential and show that the system exhibits two intertwining operators that govern the behavior of the modes.

Consider $n$ coupled classical mechanical oscillators with equal masses (set to $1$). We denote the positions of the oscillators by $\vec x = \{x_1, x_2, \dots, x_n\}^T$ such that the potential energy of the system is $\vec x^T. K .\vec x$ where $K$ is the stiffness matrix. $K$ is real and symmetric with real eigenvalues $\Omega_i^2$, and corresponding eigenvectors $\vec q_i$ satisfying
\begin{equation}
    K \vec q_i = \Omega_i^2 \vec q_i.
\end{equation}

We are interested in the first-order equation
\begin{equation}
i\frac{d}{dt} \begin{pmatrix} \vec x(t)\\ \vec p(t) \end{pmatrix} 
= -i\begin{pmatrix} 
\bigz & -\id \\ 
K &  \gamma \id 
\end{pmatrix}
\begin{pmatrix} \vec x(t) \\ \vec p(t) \end{pmatrix},
\end{equation}
where $\gamma$ is the viscous damping coefficient. The equation above also defines the quantum Hamiltonian
\begin{equation}
    H = -i\begin{pmatrix} \bigz & -\id \\ K & \gamma \id \end{pmatrix}.
\end{equation}
Since the Hamiltonian is time-independent the equation can be solved by substituting $\begin{pmatrix} \vec x (t)\\ \vec p(t) \end{pmatrix} = e^{-i\omega t} \ket{v}$, where $\ket v$ is a time-independent column vector, to get
\begin{equation}
\omega \ket{v} = -i\begin{pmatrix} \bigz & -\id \\ K & \gamma \id \end{pmatrix} \ket{v},
\end{equation}
an eigenvalue equation. The eigenvectors are
\begin{equation}\label{eq:oscillatorvectors}
\ket{v_i^{\pm}} =\begin{pmatrix} \vec q_i \\ -i\omega_{i}^{\pm} \vec q_i \end{pmatrix}
\end{equation}
with eigenvalues
\begin{equation}\label{eq:oscillatorvalues}
    \omega_i^{\pm} = -\frac{i\gamma}{2} \pm \frac{\sqrt{ 4\Omega_i^2 - \gamma^2}}{2}.
\end{equation}

We first consider the case of no dissipation, $\gamma = 0$, such that the eigenvalues are guaranteed to be real, $\omega_i^{\pm} = \pm \Omega_i$. The pseudo-Hermitian symmetries of $H(\gamma = 0)$ are as following.

First, since this is a classical mechanical system the underlying matrix is real-valued. By our choice of notation this implies $H(\gamma = 0)^* = -H(\gamma = 0)$ such that its eigenvalues are either purely imaginary or come in pairs with oppositely signed real parts. Recalling that such a condition is equivalent to (anti-)pseudo-Hermiticity we find that we can write this as $\ga ' H = - H^\dagger \ga '$, where
\begin{equation}
\ga ' =
\begin{pmatrix}
    - \tau K & \id\\
    \id & \tau \id
\end{pmatrix}
\end{equation}
is Hermitian and invertible. Here, $\tau$ is an arbitrary real number.

Second, $H(\gamma = 0)$ is also pseudo-Hermitian, $ \gb H(\gamma = 0) \gb^{-1} = H(\gamma = 0)^\dagger$ with
\begin{equation}\label{eq:symplectic}
    \gb = iJ = i\begin{pmatrix} \bigz & \id \\ -\id & \bigz \end{pmatrix},
\end{equation}
which is Hermitian as well as unitary. This additional pseudo-Hermiticity comes from the equations of motion being derived from Hamilton's equations of motion ~\cite{zhang2016Krein} ($J$ is the symplectic form). Pseudo-Hermiticity implies that eigenvalues are either real or come in complex conjugate pairs. 

Taking the two symmetries together, $H(\gamma = 0)$ either has pairs of oppositely signed eigenvalues that are purely real or purely imaginary, or it has quadruplets of eigenvalues with nonzero real as well as imaginary parts forming the set, $\{\lambda, -\lambda, \lambda^*, -\lambda^*\}$. These symmetries were noted in Ref.~\cite{Susstrunk2016phonons} where they were connected to the time-reversal symmetry and the chiral symmetry respectively.

With dissipation present, it is useful to work with the traceless matrix,
\begin{equation}
    \Tilde H = H +\frac{i\gamma}{2}\mathbb{I}_{2n} = -i\begin{pmatrix} -\frac{\gamma}{2}\id & -\id \\ K & \frac{\gamma}{2}\id \end{pmatrix},
\end{equation}
which has the same eigenvectors as in \eqnref{eq:oscillatorvectors} but eigenvalues shifted to $\Tilde \omega_i^{\pm} = \pm \frac{\sqrt{ 4\Omega_i^2 - \gamma^2}}{2}$. Essentially we have separated away the term governing the total loss of energy of the system and are now working with a matrix with balanced gain and loss. The symmetry due to $iH$ being real generalizes to $\ga \Tilde H = - \Tilde H^\dagger \ga$ with
\begin{equation}
\ga =
\begin{pmatrix}
    \gamma \id - \tau K & \id\\
    \id & \tau \id
\end{pmatrix}.
\end{equation}
$\ga$ is invertible as long as $\omega_i^\pm \tau \neq -i$ for any $\omega_i^\pm$. Since $\ga$ depends explicitly on the parameters of the system, its usefulness is limited as its number of positive and negative eigenvalues change whenever it passes through a noninvertible point. The other symmetry remains the same, $\gb \Tilde H \gb^{-1} = \Tilde H^\dagger$ with $\gb$ as in \eqnref{eq:symplectic}~\cite{yoshida2019ring}.

Since we have two independent intertwining operators (or indeed a continuous family of intertwining operators), symmetry breaking can only occur when the modes meeting each other are of opposite kind with respect to \emph{all} intertwining operators. To find the kind of each eigenspace, first let us compute
\begin{equation}
\langle v_i^{\pm}| \gb| v_i^{\pm}\rangle = 
\begin{pmatrix} \vec q_i^* & +i\omega_{i}^{\pm *} \vec q_i^* \end{pmatrix} \gb
\begin{pmatrix} \vec q_i \\ -i\omega_{i}^{\pm} \vec q_i \end{pmatrix} = 2|\vec q_i|^2\operatorname{Re} \omega_i^{\pm}.
\end{equation}
This evaluates to $0$ for the overdamped case, $4\Omega_i^2 \leq \gamma^2$. For the underdamped case, $4\Omega_i^2 > \gamma^2$, we see that $\langle v_i^{+}| \gb| v_i^{+}\rangle$ is of positive kind while $\langle v_i^{-}| \gb| v_i^{-}\rangle$ is of negative kind. Physically, the modes are distinguished by whether the momenta are lagging behind the positions or are ahead of them. 

Meanwhile,
\begin{align}
    \langle v_i^{\pm}| \ga| v_i^{\pm}\rangle &= 
    \begin{pmatrix} \vec q_i^* & +i\omega_{i}^{\pm *} \vec q_i^* \end{pmatrix} \ga
\begin{pmatrix} \vec q_i \\ -i\omega_{i}^{\pm} \vec q_i \end{pmatrix}\nonumber \\
&= \left (|\omega_i^\pm |^2 \tau - \Omega_i^2 \tau + 2 \operatorname{Im} \omega_i^\pm + \gamma \right) |\vec q_i|^2.
\end{align}
For the underdamped case, $4\Omega_i^2 > \gamma^2$, this evaluates to $0$. For the overdamped case, $4\Omega_i^2 \leq \gamma^2$, it evaluates to
\begin{equation}
    \langle v_i^{\pm}| \ga| v_i^{\pm}\rangle = \left( \frac{\gamma^2 - 4\Omega_i^2}{2} \pm \frac{\sqrt{\gamma^2 - 4\Omega_i^2}}{2}(2- \gamma \tau) \right )|\vec q_i|^2.
\end{equation}
To yield the strongest conditions we choose $\tau = 2/\gamma$ such that both the modes are of positive kind as long as $4\Omega_i^2 < \gamma^2$.

Let us see what the combination of symmetries implies. If we start from $\gamma=0$ (when $\omega_i^\pm = \pm \Omega_i$), the conditions due to $\ga$ yield that all modes are of indefinite kind and are not very useful. $\gb$, on the other hand, shows that that the positive and negative frequencies are of opposite kind. They can become complex valued only if they meet at zero modes. Indeed these zero modes (or floppy modes) govern the instability of mechanical systems~\cite{kane2014topological}. These modes have been a subject of interest since for certain lattices they arise from a topological origin and are localized at the boundary and insensitive to local perturbations~\cite{kane2014topological}.

On increasing $\gamma$, eigenvalues of opposite kind (opposite according to $\gb$) meet each other if $\Omega_i = \gamma/2$, a condition known as critical damping. The exact solution for the eigenvalues in \eqnref{eq:oscillatorvalues} shows that this pseudo-Hermiticity breaking occurs at an exceptional point. For mechanical lattices, these exceptional points can form exceptional rings due to rotational symmetry~\cite{yoshida2019ring}. On increasing $\gamma$ further, the eigenvalues are now  of positive kind with respect to $\ga$ and thus cannot wander freely in the complex plane, i.e.,  $\omega_i^\pm + i\frac{\gamma}{2}$ is constrained to be on the imaginary axis.

In summary, many properties of coupled oscillators are captured by pseudo-Hermiticity and these ideas may be useful in lattices and time-dependent systems and even nonlinear systems~\cite{weis2022nonlinear} where exact solutions are intractable.

%\subsection{Exceptional rings and nodal topology?}
 
\section{Conclusions and discussion}\label{sec:conclusion}

In summary, we showed that each eigenvalue/eigenspace of a pseudo-Hermitian matrix, $H$ such that $H = G^{-1} H^\dagger G$, can be classified into three kinds according to the sign of $\expect{v|G|v}$: positive, negative, or indefinite. Real nondegenerate eigenvalues of a parametrized matrix $H(x)$ are either of positive kind or of negative kind, and as they wander along the real axis, on the variation of the parameter, $x$, these eigenvalues can turn into exceptional point degeneracies and/or split into complex conjugate pairs if and only if they meet a real eigenvalue of opposite kind. This then enables one to predict the occurrence of exceptional points and points of pseudo-Hermiticity breaking. 

On the basis of this classification, we also showed that the parameter space of a pseudo-Hermitian matrix exhibits topologically disconnected regions where all the eigenvalues of the matrix are real --- which in many cases correspond to distinct stable phases in physical systems. These regions are surrounded by exceptional surfaces which comprise all possible real-valued EPs of pseudo-Hermitian matrices. Exceptional surfaces that are boundaries to two different regions may meet annihilating each other and giving rise to DPs. These exceptional surfaces together with the DPs created by their intersections comprise all points of pseudo-Hermiticity breaking.

We also showed how the intertwining operator, $G$, gives rise to a conserved quantity and derived the conditions for when degeneracies caused by external symmetries are susceptible to thresholdless pseudo-Hermiticity breaking. We illustrated our results with examples from different branches of physics.

The topological ideas in this paper contribute to the broader study of non-Hermitian topological phenomena such as  symmetry-protected topological phases, nodal phases~\cite{budich2019nodal, staalhammar2021exceptionalNodal}, the graph topology of spectra~\cite{tai2022graph}, etc. It would be interesting to investigate if the results of this paper can be generalized and applied to other symmetry classes of non-Hermitian matrices.  A comprehensive study of the interplay of external symmetries and pseudo-Hermiticity, and the application of this work to the study of random non-Hermitian matrices~\cite{xiao2022levelstatistics,metz2019spectral} and to time-dependent systems~\cite{chang2019bifurcation} are all interesting directions. Investigating the response strength of DPs and EPs at pseudo-Hermiticity breaking points is also an interesting direction~\cite{wiersig2022response}. We leave these questions to future work.

\section{Acknowledgements}

I thank Jayson Paulose and Wenqian Sun for reading an earlier version of the manuscript and providing feedback and suggestions. I also thank Koppara Philip Thomas for useful discussion. Finally I thank the UO Libraries and Open Access
Article Processing Charge Fund Committee for their contribution to the payment of the article processing charge. Work was
partially supported by the National Science Foundation under
Grant No. DMR-2145766.

\bibliography{references}

\appendix

\section{Proofs for statements in \secref{sec:structure}}\label{app:miniProofs}

In this section, $U = e^{-iH}$ such that $G = U^\dagger G U$ [see \eqnref{eq:eqnU}].

\begin{lemma}\label{lemma1}
Let $\lambda$ be an eigenvalue of $H$ such that $\operatorname{Im}\lambda \neq 0$ and $\ket{v}$ a corresponding eigenvector. Then, $\expect{v|G|v} = 0$.
\end{lemma}

\begin{proof} We have
\begin{align}
\expect{v|G|v} &= \expect{v|U^\dagger G U|v}\nonumber\\
&= \expect{v|e^{+i\lambda^*} G e^{-i\lambda}|v}\nonumber\\
&= e^{2 \text{Im}\lambda}\expect{v|G|v}.
\end{align}

Since, $e^{2 \text{Im}\lambda} \neq 1$ we have $\expect{v|G|v}$ = 0.
\end{proof}

\begin{lemma}\label{lemma2}
Let $\lambda$ be a real eigenvalue of $H$ with geometric multiplicity less than algebraic multiplicity (an exceptional point). Then, $\expect{v|G|v} = 0$ for some corresponding eigenvector $\ket{v}$.
\end{lemma}

\begin{proof}Since the algebraic multiplicity of $\lambda$ is greater than its geometric multiplicity, we can define at least two linearly independent vectors $\ket{v}$ and say $\ket{w}$ such that
\begin{align}
H\ket{v} = \lambda \ket{v} \quad &\text{or} \quad U\ket{v} = e^{-i\lambda} \ket{v}, \text{ and} \label{eq:exceptional1}\\
H\ket{w} = \lambda \ket{w} + \ket{v} \quad &\text{or} \quad U\ket{w} = e^{-i\lambda} \ket{w} -ie^{-i\lambda} \ket{v}. \label{eq:exceptional2}
\end{align}

$\ket{w}$ is called a generalized eigenvector and satisfies $(H-\lambda \mathbb{I})^k\ket{w} = 0$ with $k>1$. [The second part of \eqnref{eq:exceptional2} follows from $H^n\ket{w} = \lambda^n \ket{w} + n\lambda^{n-1}\ket{v}$, which can be proven by mathematical induction.] Then,
\begin{align}\label{exceptionalProof}
\expect{v|G|w} &= (\expect{v|U^\dagger) G (U|w})\nonumber\\
&= e^{+i\lambda^*}\bra{v}G(e^{-i\lambda} \ket{w} -ie^{-i\lambda} \ket{v})\nonumber\\
&= \expect{v|G|w} -i\expect{v|G|v},
\end{align}
implying $\expect{v|G|v}$ is zero.
\end{proof}

\section{Proof of Krein-Gel'fand-Lidskii Theorem}\label{app:KGL}

Let $H(\bs k)$ be a parameterized pseudo-Hermitian matrix such that $G(\bs k)H(\bs k) = H^\dagger(\bs k)G(\bs k)$. $H(\bs k)$ and $G(\bs k)$ are continuous functions of the parameter(s) $\bs k$, and $G(\bs k)$ is Hermitian and invertible for all values of $\bs k$. Note that the case of $G(\bs k) = G$ being constant is automatically covered as a special case.

\subsection{Proof of sufficiency} 

To show that if $\expect{v|G(\bs k = \bs k_0)|v}\neq 0$ for all eigenvectors associated with $\lambda_0$, a real eigenvalue of $H(\bs k = \bs k_0)$, then $\lambda_0$ stays real and diabolic upon small perturbations of $\bs k$.

\begin{proof}Suppose the contrary. Then there exists a sequence of matrices, $H_1, H_2, \dots \to H_0 = H(\bs k_0)$ and $G_1, G_2, \dots \to G_0 = G(\bs k_0)$ such that an eigenvalue $\lambda_m$ of $H_m$ has either non-zero imaginary part or algebraic multiplicity strictly more than geometric multiplicity. In either case, $\expect{v_m | G_m | v_m} = 0$ where $H_m \ket{v_m} = \lambda_m \ket{v_m}$ and $G_m H_m = H^\dagger_m G_m$.

While the sequence $\ket{v_1}, \ket{v_2}, \dots \to \ket{v_0}$ may not be convergent, we can select a subsequence such that $\ket{v_{m_j}}\to \ket{v}$ as $m_j\to\infty$. By suitably changing the notation, we have $\ket{v_m}\to\ket{v}$ as $m\to\infty$. Letting $m\to\infty$ in the above equalities, we obtain
\begin{equation}
    H_0 \ket{v} = \lambda_0 \ket{v} \quad \text{with } \expect{v|G_0|v} = 0.
\end{equation}

But we assumed $\expect{v|G_0|v}\neq 0$ for all eigenvectors associated with $\lambda_0$.
\end{proof}

This proof closely follows the one given in Chapter III of Ref.~\cite{YakubovichStarzhinskii}. An alternate proof can be found in Ref.~\cite{coppel1965stability}.

\subsection{Proof of necessity} 

To show that if $H(\bs k = \bs k_0)$ has real eigenvalues and $\expect{v|G(\bs k = \bs k_0)|v} = 0$ for some eigenvector associated to an eigenvalue $\lambda$, then it is possible to perturb $\bs k$ in such a way that $\lambda$ splits into complex conjugate eigenvalues.

\begin{proof}

We write $H(\bs k = \bs k_0) := H(0)$ in its Jordan normal form, $H(0) = P J P^{-1}$. The columns of the matrix $P$ are the generalized right eigenvectors of $H(0)$ which we denote by $\ket{R_i}$. The rows of the matrix $P^{-1}$ are the generalized left eigenvectors of $H(0)$, which we denote by $\bra{L_i}$. Evidently, $\expect{L_i|R_j} = \delta_{ij}$ since $P^{-1}P$ is the identity. We will work with this biorthonormal basis.

\textit{Case of diabolic point.} We consider the case of a diabolic point at $\lambda$ such that $H(0)$ is diagonalizable at least in the associated root subspace. This case will, in particular, also prove that an exceptional point is not necessary for symmetry breaking to occur. 

For simplicity we will assume that the degeneracy in $\lambda$ is twofold,
\begin{equation}
H(0)=\left (\lambda\ket{R_1}\bra{L_1} + \lambda\ket{R_2}\bra{L_2}\right ) \oplus \Tilde H,
\end{equation}
where $\Tilde H$ is the projection of the matrix in the rest of the space. A general Hermitian matrix can be expressed in terms of the left eigenvectors of $H$ as $M =\sum_{i,j}m_{ij}\ket{L_i}\bra{L_j}$. We claim that $G$ can be written as
\begin{align}
G&=\left (\eta_1\ket{L_1}\bra{L_1} + \eta_2\ket{L_2}\bra{L_2}\right ) \oplus \Tilde G\\
:&=G_\lambda\oplus \Tilde G,
\end{align}
where $\Tilde G$ is the projection of $G$ in the subspace spanned by $\{\ket{L_3},\ket{L_4},\dots \}$.

To show that there are no cross terms such as $\ket{L_i}\bra{L_1}$ or $\ket{L_i}\bra{L_2}$, with $i>2$, in the expression above, we need to show that $\expect{R_i|G|R_1} = \expect{R_i|G|R_2} = 0$ for $i>2$. If $\ket{R_i}$ is an eigenvector, we have
\begin{align}
\expect{R_i|G|R_1} &= \expect{R_i|U(t)^\dagger G U(t)|R_1}\nonumber\\
&= \expect{R_i|e^{+i\lambda_i^* t} G e^{-i\lambda t}|R_1}\nonumber\\
&= e^{i(\lambda_i^*-\lambda)t}\expect{R_i|G|R_1}.
\end{align}
Since $\lambda_i \neq \lambda$, $e^{i(\lambda_i^*-\lambda)t} \neq 1$ at all times $t$, and thus $\expect{R_i|G|R_1} = 0$. Similarly, $\expect{R_i|G|R_2} = 0$. If on the other hand, $\ket{R_i}$ is a generalized eigenvector we proceed similar to Lemma \ref{lemma2}.

We can now choose $\ket{L_1}$ and $\ket{L_2}$ appropriately to diagonalize $G$ in this subspace. Since $G$ is Hermitian and invertible and $\expect{R|G|R} = 0$ for some eigenvector associated with $\lambda$, we must have $\eta_1 = \expect{R_1|G|R_1}$ and $\eta_2 = \expect{R_2|G|R_2}$ being real, nonzero, and of opposite signs. 

The proof is now simply done by construction. A Hamiltonian which is pseudo-Hermitian with respect to $G_\lambda$, $H^\dagger G_\lambda = G_\lambda H$, is
\begin{align}
H &= a\ket{R_1}\bra{L_1} + \eta_2(b+ic)\ket{R_1}\bra{L_2}\nonumber\\
&+ \eta_1(b-ic)\ket{R_2}\bra{L_1} + d\ket{R_2}\bra{L_2}
\end{align}
with $a,b,c,d$ being arbitrary real constants. We now set $a=d=\lambda\cos(x), b=\frac{\lambda\sin(x)}{\sqrt{|\eta_1 \eta_2|}}, c=0$ to get
\begin{align}
    H(x) &= \lambda \bigg (\cos(x)\ket{R_1}\bra{L_1} + \eta_2\frac{\sin(x)}{\sqrt{|\eta_1 \eta_2|}}\ket{R_1}\bra{L_2}\nonumber\\
&+ \eta_1\frac{\sin(x)}{\sqrt{|\eta_1 \eta_2|}}\ket{R_2}\bra{L_1} + \cos(x)\ket{R_2}\bra{L_2}\bigg )
    \oplus \Tilde{H},
\end{align}
which has eigenvalues $\lambda e^{\pm ix}$ in the relevant subspace. These eigenvalues are complex-valued for real $x$ and at $x = 0$ we recover the original Hamiltonian, $H(0)$.

\textit{Case of exceptional point.} The proof is similar when there is an exceptional point at $\lambda$, a real eigenvalue of $H(0)$. For simplicity we will again assume that the degeneracy in $\lambda$ is twofold,
\begin{equation}
H(0)=\left (\lambda\ket{R_1}\bra{L_1} + \lambda\ket{R_2}\bra{L_2}+ \ket{R_1}\bra{L_2}\right ) \oplus \Tilde H.
\end{equation}
Here, $\ket{R_1}$ is a right eigenvector of $H(0)$ and $\ket{R_2}$ a generalized right eigenvector. 

$G$ can again be made block-diagonal in the left eigenvector basis (see the diabolic case). We now use Eqs. \eqref{eq:exceptional1} and \eqref{eq:exceptional2} to show that $\expect{R_1|G|R_2} = \expect{R_2|G|R_1}$,
\begin{align}
    \expect{R_2|G|R_2} &= \expect{R_2|U^\dagger G U |R_2}\nonumber \\
    &= \left (e^{i\lambda} \bra{R_2} +ie^{i\lambda} \bra{R_1}\right )G\left (e^{-i\lambda} \ket{R_2} -ie^{-i\lambda} \ket{R_1}\right )\nonumber\\
    &= \expect{R_2|G|R_2} + i\expect{R_1|G|R_2} - i\expect{R_2|G|R_1}.
\end{align}

Since $G_\lambda$ is Hermitian and invertible, and $\expect{R_1|G|R_1} = 0$ by Lemma \ref{lemma2}, we must have
\begin{equation}
    G_\lambda = \eta_1\ket{L_2}\bra{L_1} + \eta_1 \ket{L_1}\bra{L_2} +\eta_2 \ket{L_2}\bra{L_2},
\end{equation}
where $\eta_1 = \expect{R_1|G|R_2}$  and $\eta_2 = \expect{R_2|G|R_2}$ are real and $\eta_1$ is nonzero. A Hamiltonian which is pseudo-Hermitian with respect to $G_\lambda$, $H^\dagger G_\lambda = G_\lambda H$, is
\begin{align}
H &= \left ((a-ib)\eta_1 - c\eta_2\right )\ket{R_1}\bra{L_1} + (d\eta_1-ib\eta_2)\ket{R_1}\bra{L_2}\nonumber\\
&+ c\eta_1\ket{R_2}\bra{L_1} + (a+ib)\eta_1\ket{R_2}\bra{L_2},
\end{align}
with $a,b,c,d$ being arbitrary real constants. We now set $a=\frac{\lambda\cos(x)}{\eta_1}, b=\frac{\lambda\sin(x)}{\eta_1}, c=0,d=\frac{\cos(x)}{\eta_1}$ to get
\begin{align}
H(x) &= \lambda e^{-ix}\ket{R_1}\bra{L_1} + \bigg(\cos(x)-\frac{i\lambda\eta_2\sin(x)}{\eta_1}\bigg )\ket{R_1}\bra{L_2}\nonumber\\
&+\lambda e^{+ix}\ket{R_2}\bra{L_2},
\end{align}
which has eigenvalues $\lambda e^{\pm ix}$ in the relevant subspace.
\end{proof}

The complete proof covering $r$-fold degeneracies is similar and can be found in Chapter III of Ref.~\cite{YakubovichStarzhinskii} and in Ref.~\cite{coppel1965stability}.

\section{Continuity of the kind of eigenvalues}\label{app:continuity}

Let $H(x)$ be a parameterized pseudo-Hermitian matrix such that $G(x)H(x) = H^\dagger(x)G(x)$, where $H(x)$ and $G(x)$ are continuous functions of $x$, and $G(x)$ is Hermitian and invertible. Here we show that eigenvalues of $H(x)$, which are of positive (or negative) kind, retain their kind on the variation of the parameter $x$. For accounting purposes, an eigenvalue $\lambda$ with $\operatorname{Im} \lambda > 0$ will be considered of positive kind, and if $\operatorname{Im} \lambda < 0$, it is of negative kind. We will assume the following Lemma the proof of which can be found in Ref.~\cite{YakubovichStarzhinskii}.

\begin{lemma}
If $G(x)$ has $p$ positive eigenvalues and $q$ negative eigenvalues then $H(x)$ has $p$ eigenvalues of positive kind and $q$ eigenvalues of negative kind (after counting multiplicities). The converse is also true.
\end{lemma}

A corollary of this is the following (with the proof in Ref.~\cite{YakubovichStarzhinskii}). 

\begin{lemma}\label{lem:pq}
Let $p$ eigenvalues of $H(x)$ of positive kind and $q$ eigenvalues of negative kind meet at $\lambda$. Let $P$ be the eigenprojection operator associated with $\lambda$ (which projects any vector to the root subspace of $\lambda$)~\cite{ashida2020nonhermitian}. Then $G(x) P$ has $p$ positive eigenvalues and $q$ negative eigenvalues. The converse is also true.
\end{lemma}

We are now ready to prove the continuity of the kind of eigenvalues.

\begin{proof}

Let $\gamma_j$ describe small non-intersecting disks of radius $\epsilon$ around every \emph{distinct} eigenvalue, $\lambda_j$, of $H(0)$. We must show that for arbitrarily small $\epsilon >0$ we can find a $\delta = \delta(\epsilon)$ such that $H(\delta)$ has the same number of eigenvalues of the first kind and the same number of eigenvalues of the second kind inside each $\gamma_j$ as $H(0)$ does.

Let $\lambda$ be a $k$-fold eigenvalue of $H(0)$. If $\operatorname{Im} \lambda > 0$ we keep $\epsilon$ small enough that the corresponding disk $\gamma$ does not touch the real axis. Now, by the continuity of eigenvalues one can choose a $\delta$ small enough that $H(\delta)$ also has $k$ eigenvalues inside $\gamma$. A similar procedure works when $\operatorname{Im} \lambda < 0$.

Now we consider the case where $\lambda$ is a real $k$-fold eigenvalue of $H(0)$ where $p$ eigenvalues of first kind and $q=k-p$ eigenvalues of second kind meet. Let $\Omega(0)$ and $\Omega(\delta)$ denote the sum of the root subspaces corresponding to the eigenvalues of $H(0)$ and $H(\delta)$, respectively, which lie inside $\gamma$. Let  $P(0)$ and $P(\delta)$ denote the projection matrices corresponding to $\Omega(0)$ and $\Omega(\delta)$ respectively. 

To define the projection matrices explicitly we make use of the resolvent $R(z) = (M - z \id)^{-1} $ of a matrix, $M$. The resolvent is analytic in a region that does not contain any eigenvalues of $M$~\cite{ashida2020nonhermitian}. Since $\lambda$ is an isolated eigenvalue the projection matrices above can be related to the resolvent through~\cite{ashida2020nonhermitian,kato1966perturbation}
\begin{equation}
    P(0) = \frac{i}{2\pi} \int_{\tau}  \! dz \, \left (H(0) - z \id \right )^{-1}
\end{equation}
and
\begin{equation}
    P(\delta) = \frac{i}{2\pi} \int_{\tau} \! dz \, \left (H(\delta    ) - z \id \right )^{-1},
\end{equation}
where $\tau$ is the circumference of $\gamma$. We see that $P(0)$ can be made as close as desired to $P(\delta)$ provided $H(0)$ is sufficiently close to $H(\delta)$. By Lemma~\ref{lem:pq}, $G(0)P(0)$ has $p$ positive and $q$ negative eigenvalues. By the continuity of eigenvalues, $G(\delta)P(\delta)$ also has $p$ positive and $q$ negative eigenvalues. Thus, by Lemma~\ref{lem:pq} again, $H(\delta)$ has $p$ eigenvalues of positive kind and $q$ eigenvalues of negative kind inside $\gamma$. 
\end{proof}

This proof closely follows the one given in Chapter III of Ref.~\cite{YakubovichStarzhinskii}.

\section{Finding the intertwining operator}\label{app:finding}

For any pseudo-Hermitian matrix, $H$, the intertwining operator $G$ is not unique. Given a Hermitian matrix $G_n$ such that $H = G_n^{-1} H^\dagger G_n$ we can construct another intertwining operator, $G_{n+1} = G_n H$ which is also Hermitian~\cite{ruzsicka2021conserved}.

An exhaustive method to find all possible solutions $G$ for the equation $AG = GB$ is,
\begin{align}
    AG &= GB,\nonumber\\
    (I\otimes A)\operatorname{vec}G &= (B^T\otimes I)\operatorname{vec}G,\nonumber\\
    (I\otimes A - B^T\otimes I)\operatorname{vec}G &= 0.
\end{align}
Here $\otimes$ is the Kronecker product and $\operatorname{vec}G$ is created by arranging the entries of the matrix $G$ in a column, $(g_{11}, g_{12},\dots, g_{21},\dots)^T$. Evidently the solutions for $G$ form a vector space since $\operatorname{vec}G$ forms the null space of a matrix.

%It may be useful, for computational speed, to first attempt to diagonalize the matrices~\cite{noble2017diagonalization}.
Our results apply to every Hermitian intertwining operator that one can find for a pseudo-Hermitian matrix $H$.

\section{Formulation in terms of generalized $\mathcal{PT}$ symmetry}\label{app:PT}

Pseudo-Hermitian matrices commute with the generalized $\mathcal{PT}$ symmetry operator, $H = SH^* S^{-1}= S \mathcal{T} H \mathcal{T}^{-1} S^{-1}$ [see \eqnref{eq:defnPseudo2}]. Formulating the results of this paper in terms of $S$ and $\mathcal{T}$ is, however, not straightforward, partly due to the complications of $\mathcal{T}$ being an antilinear operator~\cite{uhlmann2016anti}.

Since we are concerned with symmetry breaking, let us assume $H$ starts off with all real eigenvalues and no exceptional point degeneracies. We can then diagonalize $H$ as
\begin{equation}
    H = \sum_i \lambda_i \ket{R_i}\bra{L_i}.
\end{equation}

The transpose of $H$ is $H^T = \sum_i \lambda_i \ket{L_i}^* \bra{R_i}^*$. Every matrix is similar to its transpose and in this case the similarity transformation is given by,
\begin{align}
    H &= K H^T K^{-1} \quad \text{where }\\
    K &= \sum_ie^{i\phi_i}\ket{R_i}\bra{R_i}^*\quad \text{and}\label{defineK}\\
    K^{-1} &= \sum_i e^{-i\phi_i}\ket{L_i}^* \bra{L_i}.
\end{align}
Here, each $e^{i\phi_i}$ is an arbitrary phase factor. Note that Eq.\eqref{defineK} is \emph{not} invariant under a change in absolute phase of any eigenvector, $\ket{R_i}\to e^{i\theta}\ket{R_i}$. Now
\begin{equation}
    H = G^{-1} H^\dagger G = G^{-1} (K^{-1})^* H^* K^* G,
\end{equation}
implying $S=G^{-1} (K^{-1})^*$. Now that we have expressed $S$ in terms of $G$ all that remains is to translate results in terms of $G$ to $S$.

We note that $\Tilde H = H^\dagger$ is pseudo-Hermitian with respect to the Hermitian matrix $G^{-1}$. Symmetry breaking of $H^\dagger$ is characterized by the sign of $\bra{\Tilde R_i} G^{-1} \ket{\Tilde R_i} =\bra{L_i} G^{-1} \ket{L_i}= \bra{L_i} S K^* \ket{L_i}$ since the right eigenvector of $H^\dagger$ is the left eigenvector of $H$. In fact, symmetry breaking in $H^\dagger$ is equivalent to it occurring in $H$ implying that the sign of $\bra{L_i} S K^* \ket{L_i}$ characterizes the kind of any eigenvector $\ket{R_i}$ of $H$. We can simplify this further,
\begin{align}
    \bra{L_i} S K^* \ket{L_i}&= e^{-i\phi_i}\bra{L_i} S \ket{R_i}^*\\
    &= e^{-i\phi_i}\bra{L_i} S \mathcal{T} \ket{R_i}.\label{emergenceT}
\end{align}
The arbitrary phase $e^{-i\phi_i}$ in the first equation above is problematic. Eq.~\eqref{emergenceT} also generates additional phase factors on the transformation $\ket{R_i}\to e^{i\theta}\ket{R_i}$ since $\mathcal{T}$ is antilinear and $\mathcal{T} e^{i\theta}\ket{R_i} = e^{-i\theta}\ket{R_i}^* \neq e^{i\theta}\ket{R_i}^*$. These complications suggest that a possible formulation in terms of $S$ and $\mathcal T$ may require phase fixing (gauge fixing) of the eigenvectors and we leave this to future work.

\section{Details for characterizing the points of symmetry breaking}

\subsection{Neighborhood of a point of symmetry breaking}\label{app:boundaries}

We start from the matrices in \eqnref{eq:originalG} and \eqnref{eq:originalH} and make an infinitesimal displacement in parameter space such that the intertwining operator changes to
\begin{equation}
    G(\epsilon) = \begin{pmatrix}
        \eta_1 + \epsilon \Tilde \eta_1 & 0\\
        0 & -\eta_2 - \epsilon \Tilde \eta_2
    \end{pmatrix} = 
    G_0 + \epsilon \begin{pmatrix}
         \Tilde \eta_1 & 0\\
        0 & -\Tilde \eta_2
    \end{pmatrix},
\end{equation}
where $\Tilde \eta_1$ and $\Tilde \eta_2$ are arbitrary real numbers and $ \epsilon > 0$ is arbitrarily small. A general pseudo-Hermitian matrix with respect to $G(\epsilon)$ and which reduces to $H_0$ when $\epsilon = 0$ is given by $H(\epsilon) =$
\begin{equation}
     \begin{pmatrix}
        \lambda + a + \epsilon (\Tilde \lambda + \Tilde a) & (\eta_2+\epsilon \Tilde \eta_2) (b + \epsilon \Tilde b) e^{i(\theta+\epsilon \Tilde \theta)}\\
        -(\eta_1+\epsilon \Tilde \eta_1) (b + \epsilon \Tilde b) e^{-i(\theta+\epsilon \Tilde \theta)} & \lambda - a + \epsilon (\Tilde \lambda - \Tilde a)
    \end{pmatrix}
\end{equation}
where $b = \frac{a}{\sqrt{\eta_1 \eta_2}}$ and $\Tilde a,\Tilde b,\Tilde \theta,\Tilde \lambda$ are real. Again we see that the eigenvalues do not depend on $\theta+\epsilon \Tilde \theta$, and are trivially shifted by $\lambda + \epsilon \Tilde \lambda$ which controls the overall trace of the matrix. We therefore set these terms to zero without loss of generality to get $H(\epsilon) =$
\begin{equation}\label{eq:loss}
    \begin{pmatrix}
         a +\epsilon \Tilde a & (\eta_2+\epsilon \Tilde \eta_2) \left (\frac{a}{\sqrt{\eta_1 \eta_2}} + \epsilon \Tilde b \right )\\
        -(\eta_1+\epsilon \Tilde \eta_1) \left (\frac{a}{\sqrt{\eta_1 \eta_2}} + \epsilon \Tilde b \right ) &  -a - \epsilon  \Tilde a
    \end{pmatrix}.
\end{equation}

We first consider the case where $a=0$ and $H_0$ had a diabolic degeneracy. In this case $H(\epsilon)$ is given by
\begin{equation}
    H(\epsilon) = \epsilon \begin{pmatrix}
        \Tilde a & (\eta_2+\epsilon \Tilde \eta_2) \Tilde b \\
        -(\eta_1+\epsilon \Tilde \eta_1)  \Tilde b & -\Tilde a
    \end{pmatrix}.
\end{equation}
This has the same form as a general pseudo-Hermitian matrix in \eqnref{eq:originalH} (with the overall trace and phase factor removed) and can therefore admit all permitted eigenvalues including from the strongly stable regions $(+,-)$ and $(-,+)$.

Now we consider the case where $a\neq0$ and $H_0$ had an exceptional degeneracy. In order to prove our statement we only need to consider perturbations which create real eigenvalues. Our strategy would be to show that the larger of these eigenvalues cannot change its kind regardless of the kind of perturbation applied.

Now for any matrix, $M = \begin{pmatrix}
    x & y\\
    z & w
\end{pmatrix}$, with real eigenvalues, the larger eigenvalue is given by $\frac{x - w + \sqrt{D}}{2}$ where the discriminant, $D = (x-w)^2 + 4yz$, is positive by the assumption of real eigenvalues. The corresponding eigenvector is $\ket{v} = \begin{pmatrix}
    x-w + \sqrt{D}\\2z
\end{pmatrix}$. The kind of this eigenspace with respect to $G(\epsilon)$ is given by the sign of
\begin{equation}
    \expect{v|G(\epsilon)|v} = |x-w + \sqrt{D}|^2(\eta_1 + \epsilon \Tilde \eta_1) - 4|z|^2(\eta_2 + \epsilon \Tilde \eta_2).
\end{equation}

On evaluating this quantity for the matrix in \eqnref{eq:loss} and expanding in orders of $\epsilon$, we find that the $\mathcal{O}(\epsilon^0)$ is zero since we started with a degeneracy of indefinite kind at $\epsilon=0$. The next term is order $\mathcal{O}(\epsilon^\frac{1}{2})$ given by $4 a \eta_1 \sqrt{D}$ (one can check that $D$ is $\mathcal{O}(\epsilon)$). Since $\eta_1 > 0$ by assumption of invertibility of $G(\epsilon)$, $a\neq 0$ by assumption of exceptional point, and $D>0$ by assumption of real eigenvalues, this is indeed the leading term with its sign being the same as the sign of $a$. 

\subsection{Topological characterization of exceptional points in the boundaries}\label{app:topological}

For the purposes of this section we will consider a constant intertwining operator
\begin{equation}
    G = \begin{pmatrix}
        +1 & 0\\
        0 & -1
    \end{pmatrix} = \sigma_3,
\end{equation}
which is Hermitian as well as unitary and follow the procedure first laid out in Ref.~\cite{yoshida2019ring}. The pseudo-Hermitian matrix (upto overall trace) is given by
\begin{equation}
    H =  \begin{pmatrix}
        a & b e^{i\theta}\\
        - b e^{-i\theta} & - a
    \end{pmatrix},
\end{equation}
where $a$ and $b$ are real. We define the Hermitian matrix
\begin{equation}
    \mathcal{H} = \begin{pmatrix}
        0 & iH\\
        -iH^\dagger & 0
    \end{pmatrix},
\end{equation}
which satisfies two Hermitian chiral symmetries, $\mathcal{H} U_1 + U_1 \mathcal{H} = 0$ and $\mathcal{H} U_2 + U_2 \mathcal{H} = 0$. Here,
\begin{equation}
    U_1 = \begin{pmatrix}
        \mathbb{I}_2 & 0\\
        0 & -\mathbb{I}_2
    \end{pmatrix}\quad \text{and} \quad 
    U_2 = \begin{pmatrix}
        0 & G\\
        G & 0
    \end{pmatrix}
\end{equation}
are both unitary. We can block diagonalize the Hamiltonian $\mathcal H$ with plus and minus sectors of $U := iU_1 U_2$. That is, we find a unitary transformation so that in the new basis,
\begin{equation}
    U = \begin{pmatrix}
        \mathbb{I}_2 & 0\\
        0 & -\mathbb{I}_2
    \end{pmatrix}\quad \text{and} \quad 
    \mathcal{H} := \begin{pmatrix}
        H_+ & 0\\
        0 & H_-
    \end{pmatrix}.
\end{equation}
This gives $H_+ = -i\sigma_2 H \sigma_1$ and $H_- = i\sigma_2 H \sigma_1$. The number of eigenvectors corresponding to negative eigenvalues of $H_+$ gives the relevant topological index---the zeroth Chern number. The eigenvalues of $H_+$ are $a\pm b$. Thus when $b^2=a^2$ (at the exceptional point) the Chern number is 0 when $a>0$ and 1 when $a<0$.

The proof here relies on the unitary of $G$. When $G$ is not unitary we may use the results from Ref.~\cite{zhou2019topology}, which provide a method to continuously deform (in a symmetry-respecting way) any invertible Hamiltonian $H$ into a unitary matrix $U$, where $U$ is given by the unitary matrix in the polar decomposition of $H = UP$ and $P = H^\dagger H$ is positive-definite. Thus, we expect these results to hold even when $G$ was originally only Hermitian and not unitary.
\end{document}